\documentclass{article}
\usepackage{amsthm}
\usepackage[left=3cm,top=3cm,bottom=3cm,right=3cm]{geometry}

 \newtheorem{thm}{Theorem}[section]
 \newtheorem{lem}[thm]{Lemma}
  
 \newtheorem{defn}[thm]{Definition}

\usepackage{color}
\usepackage{graphicx}
\usepackage{subfig}
\usepackage{url}
\usepackage{amsmath}
\usepackage{amssymb}  
\usepackage{dsfont}
\usepackage[table]{xcolor}
\usepackage{tikz}
\usepackage{arydshln}
\usepackage{array}
\usepackage{todonotes}
\usetikzlibrary{arrows,calc}

\DeclareMathOperator{\sgn}{sgn}
\date{September 25, 2013}
\newcommand{\barK}{\(\overline{k}\)}
\newcommand{\barKa}{\overline{k}}

\newcommand{\bareN}{\overline{n}}
\newcommand{\bareM}{\overline{m}}

\newcommand{\Entry}[2]{#1_{#2}}
\newcommand{\barentry}[2]{\(\overline{#1}_{#2}\)}
\newcommand{\Barentry}[2]{\overline{#1}_{#2}}
\newcommand*{\tempo}{\multicolumn{1}{c|}{0}}
\newcommand*{\tempi}{\multicolumn{1}{c|}{\vdots}}

\newcommand*{\tempb}{\multicolumn{1}{c|}{\overline{p}_{32}}}
\newcommand*{\tempbb}{\multicolumn{1}{c|}{\overline{p}_{\overline{k}2}}}
\newcommand{\matrixS}{\emph{S}}
\newcommand{\matrixR}{\emph{R}}
\newcommand{\matrixE}{\emph{E}}
\newcommand{\matrixP}{\emph{P}}
\newcommand{\matrixH}{\emph{H}}
\newcommand{\matrixSP}{\emph{S'}}
\newcommand{\matrixRP}{\emph{R'}}

\newcommand{\matrixC}{\emph{S}}
\newcommand{\matrixCP}{\emph{S'}}
\newcommand{\matrixB}{\emph{P}} 
\newcommand{\matrixSBPS}{\emph{E}} 

\newcommand{\matrixBarC}{\(\overline{S}\)} 
\newcommand{\matrixBarB}{\(\overline{P}\)} 
\newcommand{\matrixBarCe}{\overline{S}}
\newcommand{\matrixBarBe}{\overline{P}}
\newcommand{\matrixBareS}{\overline{E}}
\newcommand{\smallC}{s}

\newcommand{\BigO}[1]{\ensuremath{\operatorname{O}\bigl(#1\bigr)}}

\newcommand{\thickhline}{%
    \noalign {\ifnum 0=`}\fi \hrule height 1pt
    \futurelet \reserved@a \@xhline
}

\graphicspath{{pics/}}

\begin{document}
%
%
%

\title{On the Complexity of Reconstructing\\ Chemical Reaction Networks} 
\author{Rolf Fagerberg$^{1}$, Christoph Flamm$^2$,\\ Daniel Merkle$^{1}$, Philipp Peters$^{1}$, and Peter F. Stadler$^{2-7}$\\[0.5cm]
$^1$ Department of Mathematics and Computer Science\\ University of Southern Denmark, Denmark\\[0.2cm]
$^2$ Institute for Theoretical Chemistry, University of Vienna, Austria.\\[0.2cm]
$^3$ Bioinformatics Group, Department of Computer Science, and\\
Interdisciplinary Center for Bioinformatics, University of Leipzig, Germany.\\[0.2cm]
$^4$ Max Planck Institute for Mathematics in the Sciences, Leipzig, Germany.\\[0.2cm]
$^5$ Fraunhofer Institute for Cell Therapy and Immunology, Leipzig, Germany.\\[0.2cm]
$^6$ Center for non-coding RNA in Technology and Health\\University of Copenhagen, Denmark.\\[0.2cm]
$^7$ Santa Fe Institute, 1399 Hyde Park Rd, Santa Fe, NM 87501, USA}

\maketitle
\begin{abstract}
The analysis of the structure of chemical reaction networks is crucial for a better understanding of chemical processes. Such networks are well described as hypergraphs. However, due to the available methods, analyses regarding network properties are typically made on standard graphs derived from the full hypergraph description, e.g.\ on the so-called species and reaction graphs. However, a reconstruction of the underlying hypergraph from these graphs is not necessarily unique. In this paper, we address the problem of reconstructing a hypergraph from its species and reaction graph and show NP-completeness of the problem in its Boolean formulation. Furthermore we study the problem empirically on random and real world instances in order to investigate its computational limits in practice.
\end{abstract}


\section{Introduction}
\label{secIntro}

The use of graph models of chemical reaction networks has a long history in
physical chemistry as a means of connecting structural properties of
chemical reaction systems with the system's dynamical behavior. The aim of
this line of research to determine constraints on the dynamics in the
absence of detailed quantitative information on the reaction kinetics, as reviewed e.g.\
in \cite{Domijan:08}.

Graph models are derived either directly from a combinatorial view on the
reaction networks, i.e., from the collection of chemical reaction
equations, or from the system of ordinary differential equations that
describe the reaction kinetics. In the S-graph (also called species graph, compound
graph) for instance, we have a (directed) edge from species $v_i$ to species $v_j$ iff a chemical reaction produces output molecule $v_j$ from input molecule $v_i$. In the closely related interaction graph, we
have an edge from species $v_i$ to species $v_j$ iff $\partial [v_i]/\partial[v_j]\ne 0$, where $[v_i]$ denotes the concentration of species $v_i$ as usual in the chemical literature. Here, edges are
in addition endowed with a sign given by the partial derivative. Strong
constraints on the dynamics can be obtained if the interaction graph has no
directed cycle with an odd number of negative signs; for instance
multistability can be ruled out \cite{Soule:03,Thomas:81}.

The bipartite SR-graph (species-reaction graph) \cite{Craciun:06}
treats species and reactions as two types of vertices. In its directed
version, there is a directed edge from a species $v_i$ to a
reaction node $e$ if $v_i$ is a reactant (input molecule) in $e$, and
a directed edge from $e$ to a species $v_j$ if $v_j$ is a product (output
molecule) in $e$. This model is equivalent to the hypergraph
formulation we use in this paper (see Sec.~\ref{secForm}). The closely
related plain SR-graph is undirected and instead uses an elaborate
labeling scheme; it can be viewed as a special type of Petri-net
graphs. SR-graphs have a close relationship to classical deficiency
theory \cite{Feinberg:72,Horn:72}, the historically first approach
towards a qualitative theory of kinetics of reaction networks (see
also \cite{Shinar:13}). It can be seen as an undirected version of the
(directed) incidence graph of the directed hypergraph advocated e.g.\
in \cite{Benkoe:09a,Zeigarnik:00a}. The directed version of the
SR-Graph has been explored e.g.\ in
\cite{Ivanova:79,Mincheva:07,Volpert:87}.  These constructions have
many uses, including providing necessary conditions for Turing
instabilities in reaction-diffusion systems \cite{Mincheva:06} and
providing conditions to rule out multistability or oscillations.  A
direct relation between the directed SR-graphs and interaction graphs
that simplifies many proofs in described is \cite{Kaltenbach:12}.

A different motivation for graph models for chemical reaction systems arose
a decade ago from attempts to identify universality classes of networks
that appear in nature or have been produced by technological or social
processes \cite{Watts:98}. In this context, which focuses on connectivity
properties and path statistics as means of capturing global features,
species (S-)graphs and reaction (R-)graphs were used in particular to compare
metabolic networks \cite{Wagner:01}. Being complementary to the S-graph, the
R-graph has reactions as vertices and has an edge \((a,b)\) iff an output molecule of reaction $a$ is an input molecule of reaction $b$. 

S- and R-graph are also used for the analysis of networks, where the
directed hypergraphs, or equivalently, a directed variant of the bipartite
SR-graphs, are reformulated to S- and R- graphs. This is motivated mostly
by the available statistical toolkit employed for most other networks
models, which works on simple graphs and not hypergraphs, and the desire to
place chemical reaction networks within the scheme of small world and scale
free networks \cite{Wagner:01}. S-graphs furthermore have proved useful as
means of exploring large chemical networks \cite{Klamt:09}. They are of
practical use e.g.\ in approximation algorithms for the minimal seed set
problem \cite{Borenstein:08}, i.e., for finding the smallest set of
substrates that can generate all metabolites.

Methods for a systematic sampling of chemical networks that share
non-trivial, chemically motivated characteristic features would be
very valuable, as this would allow for statistically significant
statements for networks that appear in nature. Having an identical S-
and R-graph is one very natural characteristic feature that can be
used.

Both the S-graph and the R-graph obviously capture only partial
information on the chemical network from which they are derived. It is
less obvious, however, to what extent S-graph and R-graph together
determine the SR-graph and, equivalently, the underlying hypergraph
structure. Conversely, is there, for any given \emph{pair} of S- and
R-graphs, a chemical network from which they derive?

In this paper we address the latter question in detail by phrasing it
as a combinatorial decision problem, called the
Compound-Reaction-Reconstruction (CRR) Problem. We prove that it is
NP-complete and investigate its computational limits, using a
reformulation of the CRR problem as a Boolean satisfiability
problem.

This paper is organized as follows: In Section~\ref{secForm} we introduce
the necessary graph theoretic formalism. The definition of the CRR problem
and the NP-completeness proof are presented in Section \ref{secCRR}. In
Sec.~\ref{secDecForm} and \ref{secEmp}, we reformulate the CRR problem as a
Boolean satisfiability problem and investigate the practical applicability
and performance of different declarative solving methods.  We furthermore
consider in some detail the differences between random instances and
instances derived from known metabolic networks.


\section{Formalization of Chemical Reaction Networks}
\label{secForm}

\subsection{Chemical Reaction Networks as Directed Hypergraphs} 
\label{subMNDHG}

In this paper, a \emph{chemical reaction network} consists of set $V=\{v_1, v_2, \ldots\}$ of
molecule types (usually termed \emph{compounds} or \emph{species}
\cite{iupac1997}) and a collection of reactions \(A\). Each reaction $a$
transforms a set $a^-\subseteq V$ of reactants (possibly with multiplicities) into a
set $a^+\subseteq V$ of products (possibly with multiplicities). In chemical terms this is written
as \emph{reaction equations} of the form
\begin{equation}
   \sum_{v_i\in a^-} \alpha_{ia} v_i \to \sum_{v_i\in a^+} \beta_{ia} v_i 
\end{equation}
where the \emph{stoichiometric coefficients} $\alpha_{ia}$ and $\beta_{ia}$
are non-negative integers giving the multiplicity with which a compound $v_i \in V$ appears as reactant or
product in the given reaction $a$. We note that in general $a^-\cap a^+$ does not need
to be empty. The molecules in the intersection of products and reactants are
the so-called catalysts of the reaction. Each reaction $(a^-,a^+)$ can be
interpreted as a directed edge in a directed hypergraph.

\begin{figure}[t]
\centering
\begin{tikzpicture}[->,>=stealth',shorten >=1pt,auto,node distance=3cm,
  thick,species node/.style={draw,circle,font=\sffamily\Large\bfseries}]

\tikzstyle{reaction}=[draw,font=\sffamily\Large\bfseries]
\node[species node](E1){$A$} ;
\node[species node](E2) at ($(E1)+(0,-1.7)$){$B$} ; 
\node[species node](P1) at ($(E1)+(5,0)$){$C$} ;   
\node[species node](P2) at ($(P1)+(0,-1.7)$){$D$} ; 
\node[species node](P3) at ($(P1)+(5,0)$){$E$} ;   
\node[species node](P4) at ($(P3)+(0,-1.7)$){$F$} ; 
\draw(E1) [bend right] to node{R1} (P1);
\draw(E2) [bend left] to node{} (P2);
\draw(P1) [bend right] to node{R2} (P3);
\draw($(8.5,-0.7)$) [bend left] to node{} (P4);
\end{tikzpicture}
\caption[Example of a directed hypergraph]{A directed hypergraph, consisting of two reactions \{R1, R2\} and
  six chemical species \{A, B, $\dots$, F\}. Hyperarc R1 has $t(R1) =
  \{A,B\}$ and $h(R1)= \{C,D\}$, hyperarc R2 has $t(R2) =
  \{C\}$ and $h(R2) = \{E,F\}$.}
\label{DHG1}
\end{figure}
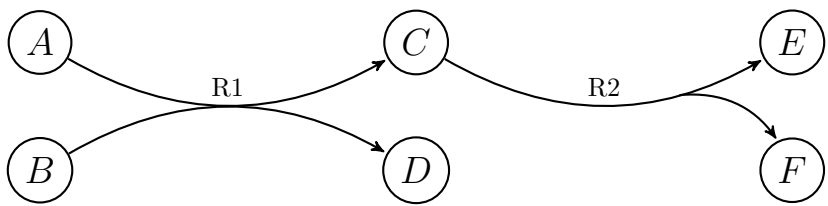

More formally, a \emph{directed hypergraph} \matrixH~is a pair $H = (V,A)$,
with a set of vertices $V$ and a set of hyperarcs $A$, where each hyperarc
$a \in A$ is an ordered pair $(t(a),h(a))$. The tail $t(a)=a^-$ of the
hyperarc in our setting refers to the reactants, while its head $h(a)=a^+$
identifies the products \cite{gallo1993,andersen2011maximizing}.
Fig.~\ref{DHG1} illustrates the hypergraph of a small chemical reaction
network with two reactions.

\subsection{Hypergraphs as Matrices}
\label{subDHGmat}

The stoichiometric matrix $\boldsymbol{\Sigma}$ with entries 
\begin{equation}
  \sigma_{ia} = \beta_{ia} - \alpha_{ia}
\end{equation}
provides a complete description of the mass balance of the each reaction in
the chemical reaction network. Each row of the matrix $\boldsymbol{\Sigma}$ corresponds to a species,
while each column is identified with a reaction. The stoichiometric matrix
is a complete encoding of the chemical reaction network (hence of the directed hypergraph),
provided $t(a)\cap h(a)=\emptyset$ for all
reactions $a\in A$. Note that $\sigma_{ia}<0$ if $v_i$ is consumed by reaction
$a$ while $\sigma_{ia}>0$ if $v_i$ is produced. 

In the following, we focus on the topological structure of the
chemical reaction network and will ignore the multiplicities of
reactants and products. This corresponds to replacing $\sigma_{ia}$ by
its sign in $\boldsymbol{\Sigma}$. Instead of using this reduced
version of $\boldsymbol{\Sigma}$, however, it will be more convenient
to use two \emph{binary} incidence matrices, $\mathbf{E}$ and
$\mathbf{P}$, defined by
\begin{equation}
\begin{split}
  e_{ia} & = \begin{cases} 1 & \text{iff } v_i \in t(a)\\
                            0 & \text{otherwise}  \\
              \end{cases}\\
  p_{ai} & = \begin{cases} 1 & \text{iff } v_i \in h(a)\\
                            0 & \text{otherwise}  \\
              \end{cases}\\ 
\end{split}
\end{equation}
Here, \(e_{ia}\) is \( \sgn \alpha_{ia}\) and \(p_{ai}\) is \( \sgn \beta_{ia}\). The rows of the $n \times m$ matrix \matrixE\ correspond to the reactants of each reaction, while the columns of the $m \times n$ matrix \matrixP\ corresponds to the products.
The matrices  \matrixE\ and \matrixP\ corresponding to the hypergraph of
Fig.~\ref{DHG1} are shown as part of Fig.~\ref{totalG}.

In the S-graph (species graph), two species $v_i, v_j \in V$ are adjacent 
iff there is a reaction that has species $v_i$ as reactant and species
$v_j$ as product. Correspondingly, the R-graph (reaction graph) two
reactions $a$ and $b$ are adjacent iff there is a species $v$ that is 
a product of $a$ and a reactant of $b$. In other words, the adjacency matrices $S= \left( s_{ij} \right)$ and $R= \left( r_{ab} \right)$ are given by:
\begin{equation}
\begin{split}
   s_{ij} &= \begin{cases} 1& \text{iff } 
     \exists a \in A: v_i \in t(a) \wedge v_j \in h(a)\\
                           0& \text{otherwise} 
             \end{cases}\\
   r_{ab} &= \begin{cases} 1& \text{iff } 
     \exists v \in V: v \in h(a) \wedge v \in t(b)\\
                           0& \text{else}
             \end{cases}\\
\end{split}
\end{equation}

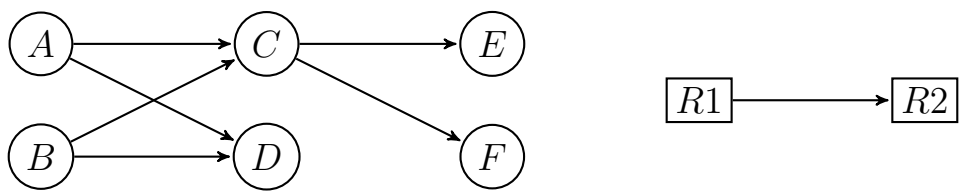
\begin{figure}[h!]
\centering
\begin{tikzpicture}[->,>=stealth',shorten >=1pt,auto,node distance=3cm,
  thick,species node/.style={draw,circle,font=\sffamily\Large\bfseries}]

\tikzstyle{reaction}=[draw,font=\sffamily\Large\bfseries]

\begin{scope}
\node[species node](A){$A$} ;
\node[species node](B) at ($(A)+(0,-1.5)$){$B$} ; 
\node[species node](C) at ($(A)+(3,0)$){$C$} ;   
\node[species node](D) at ($(C)+(0,-1.5)$){$D$} ; 
\node[species node](E) at ($(C)+(3,0)$){$E$} ;   
\node[species node](F) at ($(E)+(0,-1.5)$){$F$} ; 
\draw(A) to node{} (C);
\draw(A) to node{} (D);
\draw(B) to node{} (C);
\draw(B) to node{} (D);
\draw(C) to node{} (E);
\draw(C) to node{} (F);
\end{scope}

\begin{scope}[xshift=6.25cm,yshift=-0.75cm]
\node[reaction](R1) at ($(2.5,0)$){$R1$} ;
\node[ reaction](R2) at ($(R1)+(3,0)$){$R2$} ;   
\draw(R1) to node{} (R2);
\end{scope}
\end{tikzpicture}

\caption[Example of \matrixS- and \matrixR-graph]{The directed hypergraph from Fig.~\ref{DHG1} has the \matrixS-graph shown on the left and the \matrixR-graph, shown on the right.}
\label{SR}
\end{figure}

\subsection{Relationships Between \matrixS, \matrixR, \matrixE, and \matrixP}
\label{submatmat}

The two incidence matrices \matrixE\ and \matrixP\ together contain all
necessary information to uniquely define the corresponding hypergraph
\matrixH.  In contrast, the adjacency matrices \matrixS\ and \matrixR\
do not determine \matrixH\ uniquely, in the sense that two different
hypergraphs can lead to the same \matrixS\ and \matrixR, as illustrated
in Fig.~\ref{ambiguousSR}. Also, for a given pair of matrices \matrixSP\
and \matrixRP, there may exist no hypergraph having \matrixSP\ and
\matrixRP\ as its species and reaction graph. We provide an example of
this in the appendix. These observations lead to the question we study
in the remaining sections of this paper.

\begin{figure}[h!]
\centering
\begin{tikzpicture}[->,>=stealth',shorten >=1pt,auto,node distance=3cm,
important line/.style={thick,-},
  thick,species node/.style={draw,circle,font=\sffamily\Large\bfseries}]

\tikzstyle{reaction}=[draw,font=\sffamily\Large\bfseries]

\begin{scope}
\node[species node](A){$A$} ;
\node[ reaction](R1)at ($(A)+(2,0)$){$R1$} ;
\node[ reaction](R2) at ($(R1)+(0,-1.5)$){$R2$} ; 
\node[species node](B) at ($(A)+(0,-1.5)$){$B$} ; 
\node[species node](C) at ($(R1)+(2,0)$){$C$} ;   
\node[species node](D) at ($(C)+(0,-1.5)$){$D$} ; 
\draw(A) to node{} (R1);
\draw(A) to node{} (R2);
\draw(B) to node{} (R1);
\draw(B) to node{} (R2);
\draw(R2) to node{} (D);
\draw(R1) to node{} (C);
\end{scope}

\begin{scope}[yshift=-3cm]
\node[species node](A){$A$} ;
\node[ reaction](R1)at ($(A)+(2,0)$){$R1$} ;
\node[ reaction](R2) at ($(R1)+(0,-1.5)$){$R2$} ; 
\node[species node](B) at ($(A)+(0,-1.5)$){$B$} ; 
\node[species node](C) at ($(R1)+(2,0)$){$C$} ;   
\node[species node](D) at ($(C)+(0,-1.5)$){$D$} ; 
\draw(A) to node{} (R1);
\draw(B) to node{} (R2);
\draw(R2) to node{} (D);
\draw(R1) to node{} (C);
\draw(R2) to node{} (C);
\draw(R1) to node{} (D);
\end{scope}
\draw[important line](5,0.3) --  (5,-4.8);
 \draw[important line](-0.3,-2.25) --  (5,-2.25);
\begin{scope}[xshift=6cm]
\node[species node](A){$A$} ;
\node[species node](B) at ($(A)+(0,-1.5)$){$B$} ; 
\node[species node](C) at ($(A)+(3,0)$){$C$} ;   
\node[species node](D) at ($(C)+(0,-1.5)$){$D$} ; 
\draw(A) to node{} (C);
\draw(A) to node{} (D);
\draw(B) to node{} (C);
\draw(B) to node{} (D);
\node[ reaction](R1)at ($(A)+(1.5,-3)$){$R1$} ;
\node[ reaction](R2) at ($(R1)+(0,-1.5)$){$R2$} ; 
\end{scope}

\end{tikzpicture}

\caption[Ambiguous \matrixS- and \matrixR-graphs]{Ambiguous \matrixS- and \matrixR-graphs: On
the left, two different hypergraphs are shown, both giving rise to the
same \matrixS\ and \matrixR\ shown on the right.}
\label{ambiguousSR}
\end{figure}
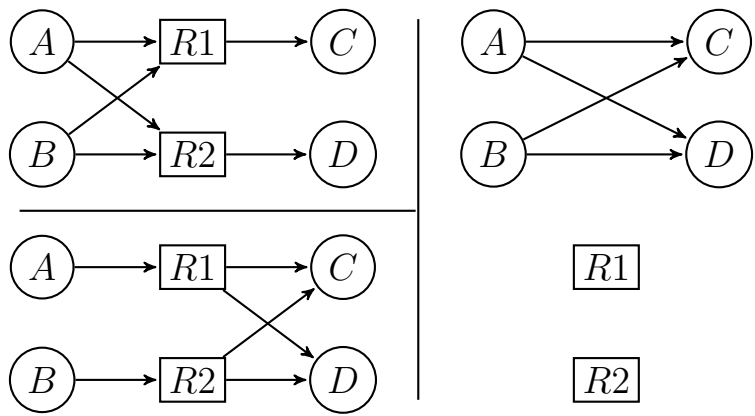

The four matrices \matrixS, \matrixR, \matrixE, and \matrixP\ can be 
combined in a natural way to the \emph{total graph} $T$ \cite{capobianco1978,skiena1990}. It corresponds to a 
\((n+m) \times (n+m)\) binary matrix of which \matrixS, \matrixR, \matrixE,
and \matrixP\ form four blocks:
\begin{equation}
    T = \left( \begin{matrix} \matrixS & \matrixE \\
                              \matrixP & \matrixR \\
               \end{matrix} \right)
\end{equation}
Formally, the total graph has vertex set $V\cup A$ and has edges defined as follows: 
$(i,a)$ is an edge in the total graph iff $i\in t(a)$ (i.e., iff $e_{ia}=1$). 
Correspondingly, $(b,j)$ is an edge iff $j\in h(b)$ (i.e., iff $p_{bj}=1$). 
For pairs of vertices of $H$ we have $(i,j)\in E(T)$ 
iff  $i,j$ are adjacent in the species graph, i.e., $s_{ij}=1$. For 
pairs of arcs of $H$, finally, $(a,b)\in E(T)$ iff $a$,$b$ are adjacent in the reaction graph, i.e., $r_{ab}=1$. 
An illustration of the total graph of  the hypergraph of Fig.~\ref{DHG1} is
displayed in Fig.~\ref{totalG}.

\begin{figure}[h!]
  \centering
\caption{The total graph of the hypergraph from Fig.~\ref{DHG1}. All non-zero entries in column $k$ in \matrixE\ are reactants of reaction $k$. All non-zero entries in row $k$ in \matrixP\ are products of reaction $k$. An entry $s_{ij}=1$ (or $r_{ab}=1$, resp.) requires at least one witness, which is R1 for the framed entry in the table, as indicated by the gray marked entries in \matrixE\ and \matrixP.}
\(   T=  \left( \begin{matrix} \matrixS & \matrixE \\
                              \matrixP & \matrixR \\
               \end{matrix} \right) = \)
  \begin{tabular}{c c c c c c c|c c c }
  &&&&&&& R1 &&\\
  &&&&&&&\(\downarrow\)&\\
  &0 &0 &1 &1 &0 &0 &1 &0\\ \cline{4-4}
  &0 &0 &\multicolumn{1}{|c|}{1} &1 &0 &0 &\cellcolor[gray]{.8}1 &0 \\ \cline{4-4}
  &0 &0 &0 &0 &1 &1 &0 &1 \\
  &0 &0 &0 &0 &0 &0 &0 &0 \\
  &0 &0 &0 &0 &0 &0 &0 &0  \\
  &0 &0 &0 &0 &0 &0 &0 &0\\
  \hline
   R1 \(\rightarrow\)&0&0&\cellcolor[gray]{.8}1&1&0&0&0 &1  \\
  &0&0&0&0&1&1&0&0 \\
  \end{tabular}
\label{totalG}
\end{figure}

The definitions of the S- and R-matrices can be recast in terms of
\matrixE~and~\matrixP. For \matrixS\ we have $s_{ij}=1 \iff \exists a:
(e_{ia}=1 \wedge p_{aj}=1)$. The reaction $a$ for which the conjunction
$(e_{ia}\wedge p_{aj})$ is true is called a \emph{witness}
\(w_{ija}^S\) \cite{alon1992}. For every non-zero entry in \matrixS\
there must be at least one witness, while no witness exists for the zero
entries. For illustration, a witness for the framed entry in \matrixS\
in Fig.~\ref{totalG} is reaction R1, and the conjunction-fulfilling
entries of \matrixE\ and \matrixP\ are marked gray. For \matrixR,
respectively, we have $r_{ab}=1 \iff \exists i: (e_{ib}=1 \wedge
p_{ai}=1)$ and the species $v_i$ called a witness \(w_{abi}^R\).

Reinterpreting the matrix elements over the standard Boolean Algebra
$(\{0,1\},\vee,\wedge)$ allows us to rewrite the definition of $s_{ij}$
in the follow form:
\begin{equation}
\label{eqS}
  s_{ij} = \bigvee_{a\in A} (e_{ia}\wedge p_{aj})
\end{equation}
Analogously, we have
\begin{equation}
\label{eqR}
  r_{ab} = \bigvee_{i\in V} (p_{ai}\wedge e_{ib})
\end{equation}
for the entries of matrix \matrixR. Using matrix notation, we can 
write \matrixS\ and \matrixR\ as Boolean matrix products 
\begin{equation}
\label{crrp}
  \matrixS = \matrixE \cdot \matrixP  \qquad 
  \matrixR = \matrixP \cdot \matrixE 
\end{equation}
where the matrix multiplication $\cdot$ uses the Boolean operations 
$\vee$ and $\wedge$ as addition and multiplication operations.

\section{The Compound-Reaction-Reconstruction (CRR(\matrixS,\matrixR)) Problem}
\label{secCRR}

Given $H$, or more precisely, \matrixE\ and \matrixP, we have seen above
that $\matrixS$ and $\matrixR$ are uniquely determined.  Here we ask the
converse question. Given $\matrixS$ and $\matrixR$, is there a hypergraph
$H$ that has $\matrixS$ and $\matrixR$ as its S- and R-matrices,
respectively. We call this problem the
\emph{Compound-Reaction-Reconstruction (CRR(\matrixS,\matrixR)) problem}.
 
\subsection{Problem Definition}
\label{subProb}

\begin{defn}
\textbf{\emph{CRR(\(\boldsymbol{\matrixC,\matrixR}\))} problem}:
Given a Boolean $n\times n$ matrix \matrixC\ and a Boolean $k\times k$ matrix 
\matrixR, is there a Boolean $n\times k$ matrix \matrixSBPS\ and a 
Boolean  $k\times n$ matrix \matrixB\ so that 
$\matrixC = \matrixSBPS \cdot \matrixB$ and  
$\matrixR = \matrixB \cdot \matrixSBPS$ ? 
\label{defCRR}
\end{defn}

The CRR(\matrixS,\matrixR) problem can be seen as a special kind of matrix
decomposition problem. Matrix decomposition methods have been studied
extensively, but mostly on real- and integer-valued matrices, as Singular
value, QR- or LU-factorization. In 2008, Miettinen et al. showed that
simple Boolean matrix decomposition (i.e., the first half of
Eqn.~(\ref{crrp})) is NP-hard \cite{miett2008}.

\subsection{NP-Completeness}
In this section, we study the complexity of the CRR(\matrixC,\matrixR)
problem. A problem is NP-complete if it is in NP and it is NP-hard, see \cite{gareyjohn1979}.
We will prove NP-completeness of the CRR(\matrixC,\matrixR)
problem by a reduction from the NP-complete Set Basis (SB) Problem
\cite[problem SP7]{gareyjohn1979}, defined in Def.~\ref{defSBP} below.

\begin{defn}
  \textbf{\emph{SB(\(\boldsymbol{\mathcal{S},k}\)) Problem}}: Input: A
  collection \(\mathcal{S}\) of subsets of a finite set \(\mathcal{E}\) and
  a positive integer \(k \leq |\mathcal{S} |\). Is there a collection
  \(\mathcal{P}\) of subsets of \(\mathcal{E}\) with \(|\mathcal{P} | = k\)
  such that, for each \(\smallC \in \mathcal{S}\), there is a subcollection
  of \(\mathcal{P}\) whose union is \(\smallC\)?
  \label{defSBP}
\end{defn}

Note that in contrast to the SB Problem, which has to satisfy one equation, the CRR(\matrixC,\matrixR) problem has to fulfill the two equations of (\ref{crrp}), i.e., entries of \matrixE\ and \matrixP\ are dependent on \matrixS\ and \matrixR, making it a twofold Boolean matrix decomposition. 
The SB(\(\mathcal{S},k\)) problem can be formulated as Boolean matrix decomposition as it was noted in \cite{miett2008}, where the subset collection \(\mathcal{S}\) is represented as a matrix \matrixC\ (i.e., each Boolean entry \(\smallC_{ij}=1\) denotes that the $i$-th subset contains the $j$-th element). Hence, \matrixP\ is the basis vector matrix and \matrixE\ is the usage matrix. Thus, the matrix decomposition formulation (SB(\matrixC,$k$) Problem) can be given as follows: Given an \(n \times m\) matrix \matrixC\ and a \(k \in \mathbb{N^+}\), do an \(n \times k\) matrix \matrixSBPS\ and a \(k \times m\) matrix \matrixB\ exist for which the following hold:
\[\boldsymbol{\matrixC = \matrixSBPS \cdot \matrixB}\]
We will now show how to transform any instance of the SB(\matrixC,\(k\)) problem into an instance of the CRR(\matrixC,\matrixR) problem. The transformation will be via an intermediate problem called SBMOD. We first modify matrix \matrixC\ to an \(\bareN \times \bareM\) matrix \matrixBarC, with \(\bareN=n+2\), \(\bareM=m+2\), and \(\barKa = k+2\), as follows:
\begin{equation*}
\matrixBarCe=
\left(\begin{array}{cccc}
1& \cdots&\cdots &1\\
\vdots&0&\cdots &0\\ \cline{3-4}
\vdots &\tempi &\begin{array}{cc}
&\\
&\matrixC  \\
\end{array}\\
1&\tempo& \\
\end{array}\right)
\end{equation*}
With the extended matrix \matrixBarC, we define a modified version of the SB(\matrixC,k) problem:

\begin{defn}
\textbf{\emph{SBMOD(\(\boldsymbol{\matrixBarCe,\barKa}\)) Problem}}: 
Given an \(\bareN \times \bareM\) matrix \matrixBarC\ in the form above and a \(\barKa \in \mathbb{N^+}\),  do an \(\bareN \times \barKa\) matrix \(\matrixBareS\) and a \(\barKa \times \bareM\) matrix \matrixBarB\ exist for which \matrixBarC = \(\matrixBareS\)~$\cdot$ \matrixBarB ?
\end{defn}

\begin{lem}
\label{c1} 
If \(\matrixBareS\)~and \matrixBarB~exist for an instance of SBMOD(\(\matrixBarCe,\barKa\)), then we can assume they have the following structure:

\begin{equation*}
\hspace{-0.7cm}
\matrixBareS =
\left(\begin{array}{cccc}
1& \cdots&\cdots &1\\
1&0&\cdots &0\\ \cline{3-4}
\Barentry{e}{31} &\tempo &\begin{array}{c}  \\\end{array}\\
\vdots&\tempi& \begin{array}{c}~~~\matrixSBPS  \\\end{array}\\
\Barentry{e}{\overline{n}1}&\tempo& \begin{array}{c} \\\end{array}\\
\end{array}\right), 
\matrixBarBe =
\left(\begin{array}{ccccc}
1& 0&\cdots &\cdots&0\\
\vdots&1& \Barentry{p}{23} &\cdots&\Barentry{p}{2\overline{m}}\\ \cline{3-5}
\vdots & \tempb&\begin{array}{c}\\
\end{array}\\
\vdots&\tempi&& \begin{array}{c}
\matrixB\\
\end{array}\\
1&\tempbb& \begin{array}{c}\\
\end{array}&\\
\end{array}\right)
\end{equation*}
where the entries \barentry{e}{ij} and \barentry{p}{kl} are not specified.

\end{lem}
\begin{proof}

Each row of \matrixBarC~is a binary linear combination of rows in \matrixBarB. Therefore, to obtain the second row of \matrixBarC, there must exist the row \(\left( 1,0,\dots,0\right) \) in \matrixBarB. Note that permuting columns of matrix \(\matrixBareS\) concurrently with rows of \matrixBarB~will not change the result of \(\matrixBareS\)~$\cdot$ \matrixBarB. We can therefore permute the columns of \(\matrixBareS\)~and the rows of \matrixBarB~concurrently, such that \(\left( 1,0,\dots,0\right) \) becomes the first row of  \matrixBarB, without changing \(\matrixBareS\)~$\cdot$ \matrixBarB. The second row of \(\matrixBareS\)~are the coefficients for the linear combination for the second row of \matrixBarC. We can choose it to be  \(\left( 1,0,\dots,0\right) \), again without changing \(\matrixBareS\)~$\cdot$ \matrixBarB.\\
The first row of \(\matrixBareS\)~are the coefficients for the linear combination of rows in \matrixBarB~giving the first row of \matrixBarC~(which are all 1's). Converting any zeros to ones does not change \matrixBarC, therefore we choose the first row of \(\matrixBareS\)~to be  \(\left( 1,\dots,1\right) \).\\
As the first row of \matrixBarC~is 1's only, we can choose the first column of \matrixBarB~to be \(\left( 1,\dots,1\right) \).\\
Due to the first row of \matrixBarC, there must exist a row of \matrixBarB, where \(\Barentry{p}{i2}=1\). It is not in the first row by the choice of the first row of \matrixBarB~just before. We permute rows 2 to \(\barKa\) in \matrixBarB~concurrently with columns 2 to \(\barKa\) in \(\matrixBareS\)~such that entry \(\Barentry{p}{22}=1\).\\
From \(\Barentry{s}{i2} = 0\) and \(\Barentry{p}{22} = 1\) follows, that \(\Barentry{e}{i2}=0\) for \(i\geq 3\).
\end{proof}

\begin{lem}
\label{c2} 
The SBMOD(\(\matrixBarCe,\barKa\)) problem has a solution iff the SB(\(\mathcal{S},k\)) problem has a solution.
\end{lem}
\begin{proof}
Forward direction: Let \matrixC, \(k\) be given, and assume the SBMOD(\matrixBarC,\barK) problem has a solution \(\matrixBareS\), \matrixBarB. By Lemma~\ref{c1}, we can assume the solution \(\matrixBareS\), \matrixBarB\ to SBMOD(\matrixBarC,\barK) has the following form:

\begin{align*}
\hspace{-0.7cm}
\left(\begin{array}{cccc}
1& \cdots&\cdots &1\\
1&0&\cdots &0\\ \cline{3-4}
\Barentry{e}{31} &\tempo &\begin{array}{c}  \\\end{array}\\
\vdots&\tempi& \begin{array}{c}~~~\matrixSBPS  \\\end{array}\\
\Barentry{e}{\overline{n}1}&\tempo& \begin{array}{c} \\\end{array}\\
\end{array}\right)
&\cdot&
\left(\begin{array}{ccccc}
1& 0&\cdots &\cdots&0\\
\vdots&1& \Barentry{p}{23}&\cdots&\Barentry{p}{2\overline{m}}\\ \cline{3-5}
\vdots & \tempb&\begin{array}{c}\\
\end{array}\\
\vdots&\tempi&& \begin{array}{c}
\matrixB\\
\end{array}\\
1&\tempbb& \begin{array}{c}\\
\end{array}&\\
\end{array}\right)
&=&
\left(\begin{array}{cccc}
1& \cdots&\cdots &1\\
\vdots&0&\cdots &0\\ \cline{3-4}
\vdots &\tempi &\begin{array}{cc}
&\\
&\matrixC  \\
\end{array}\\
1&\tempo& \\
\end{array}\right)\\
\matrixBareS\qquad\qquad\quad &\cdot&\matrixBarBe \qquad\qquad\qquad&=&\matrixBarCe \qquad\qquad\quad\\
\end{align*}
For \barentry{s}{ij} with \(i,j \geq 3\) it holds: 
\[ \Barentry{s}{ij} =  \Barentry{e}{i1} \cdot 0 + 0 \cdot \Barentry{p}{2j} + \sum\limits_{l\geq 3}^{\overline{k}} \Barentry{e}{il} \cdot \Barentry{p}{lj} = \Entry{s}{i-2,j-2} \]
Hence $\matrixE \cdot \matrixP = \matrixS$, implying that the SB(\(\mathcal{S},k\)) problem has a solution $\matrixE, \matrixP$.

Backward direction: Let \matrixC, \(k\) be given, and assume that the SB(\(\mathcal{S},k\)) problem has a solution $\matrixE, \matrixP$. A trivial solution to the SBMOD(\(\matrixBarCe,\barKa\)) problem arises from specifying as follows the entries which were left undetermined in Lemma~\ref{c1}: \(\Barentry{e}{31} = \dots=\Barentry{e}{\overline{n}1}=1\), \(\Barentry{p}{32} = \dots=\Barentry{p}{\overline{k}2}=0\), \(\Barentry{p}{23} = \dots=\Barentry{p}{2\overline{m}}=0\). This results in:

\begin{equation*}
\hspace{-0.7cm}
\left(\begin{array}{cccc}
1& \cdots&\cdots &1\\
1&0&\cdots &0\\ \cline{3-4}
\vdots&\tempi& \begin{array}{c}~~~\matrixSBPS  \\\end{array}\\
1&\tempo& \begin{array}{c} \\\end{array}\\
\end{array}\right)
\cdot
\left(\begin{array}{cccc}
1& 0&\cdots &0\\
\vdots&1& \cdots&1\\ \cline{3-4}
\vdots & \tempo&\begin{array}{cc}
&\\
& \matrixB \\
\end{array}\\
\vdots&\tempi&& \\
1&\tempo&& \\
\end{array}\right)
=
\left(\begin{array}{cccc}
1& \cdots&\cdots &1\\
\vdots&0&\cdots &0\\ \cline{3-4}
\vdots &\tempi &\begin{array}{cc}
&\\
&\matrixC  \\
\end{array}\\
1&\tempo& \\
\end{array}\right)
\end{equation*}
\end{proof}

\begin{lem} The SB(\matrixC,\(k\)) problem is NP-complete even if \matrixC\
  is a square matrix. The same applies to the SBMOD(\matrixC,\(k\))
  problem.
\end{lem}
\begin{proof} 
Showing that the two problems are in NP is done by a simple guess-and-check argument, verifying the solution in polynomial time.
  We now prove the NP-hardness. Assuming \(n < m\), \matrixC\ can easily be extended by \(m-n\) 0-rows to
  gain a square matrix \matrixCP. A solution for SB(\matrixCP,\(k\))
  consists of a \((n+(m-n) )\times k = m \times k\) matrix
  \matrixSBPS\emph{'} and a \(k \times m\) matrix \matrixB\ for which
  \(\matrixSBPS\emph{'} \cdot \matrixB = \matrixCP\). It is easy to see
  that such \matrixSBPS\emph{'}, \matrixB\ exist iff there exist an \(n
  \times k\) matrix \matrixSBPS\ and a \(k \times m\) matrix \matrixB\ such
  that \(\matrixSBPS \cdot \matrixB = \matrixC\) (in one direction, let
  \matrixSBPS\emph{'} be \matrixSBPS\ extended with \(m-n\) 0-rows, in the
  other let \matrixSBPS\ be \matrixSBPS\emph{'} with last \(m-n\) rows
  removed). Following from Lemma~\ref{c2}, the same holds for the
  SBMOD(\matrixC,\(k\)) problem. The same holds for \(n>m\) by extending
  \matrixC\ by \(n-m\) 0-columns and an analogous construction of \matrixB\emph{'}.
\end{proof}

\begin{thm} 
The CRR(\matrixC,\matrixR) problem is in NP-complete.
 \end{thm} 
 
\begin{proof} 
  To show that the CRR(\matrixC,\matrixR) problem is
  in NP is done by the simple guess-and-check argument (given \matrixE\ and
  \matrixP, it can be checked in polynomial time if they fulfill
  Eq.~\ref{crrp}). For the NP-hardness we give a polynomial time reduction from the SBMOD(\matrixC,\(k\)) problem to the CRR(\matrixC,\matrixR) problem.
Let SBMOD(\matrixC,\(k\)) be an instance with a squared matrix
  \matrixC. Convert it into an instance of the CRR(\matrixC,\matrixR)
  problem where \(\matrixR=(1)^{k \times k}\), i.e., \matrixR\ is the \(k
  \times k\) matrix with all entries equal to 1.

  Forward direction: If SBMOD(\matrixC,\(k\)) has a solution, then 
  CRR(\matrixC,\((1)^{k \times k}\)) has a solution.\\
  To see this, let \matrixSBPS, \matrixB\ be a solution for
  SBMOD(\matrixC,\(k\)), i.e.\ \(\matrixC=\matrixSBPS \cdot \matrixB\). By
  Lemma \ref{c1} we can assume that all entries in the first row of
  \matrixSBPS~are equal to 1. The same applies to the structure of
  \matrixB, where all entries of the first column are equal to 1. Therefore
  \(\matrixB \cdot \matrixSBPS = (1)^{k \times k}\).

  Backward direction: If CRR(\matrixC,\((1)^{k \times k}\)) has a solution, 
  then SBMOD(\matrixC,\(k\)) has a solution.\\
  To see this, let CRR(\matrixC,\((1)^{k \times k}\)) have a solution
  \matrixSBPS, \matrixB\ with \(\matrixC = \matrixSBPS \cdot \matrixB\) and
  \((1)^{k \times k} = \matrixB \cdot \matrixSBPS\).  The second equation
  of CRR(\matrixC,\((1)^{k \times k}\)) (\(\matrixR = \matrixB \cdot
  \matrixSBPS\)) forces \matrixB\ to have k rows and \matrixSBPS~to have k
  columns. By the first equation of CRR(\matrixC,\((1)^{k \times k}\))
  (\(\matrixC = \matrixSBPS \cdot \matrixB\)) it follows that \matrixSBPS,
  \matrixB\ are a solution to SBMOD(\matrixC,\(k\)).
\end{proof}


\section{Declarative Formulations}
\label{secDecForm}

Despite its NP-completeness, the CRR problem is relevant to chemistry. We therefore want to empirically investigate the CRR problem and examine for how large networks the CRR problem can be solved in a reasonable time. To this end, we formulate it as a Boolean satisfiability problem, allowing us to use established declarative approaches, such as Satisfiability- (SAT-) solvers, Satisfiability-Modulo-Theory- (SMT-) solvers and solvers for Integer Linear Programming (ILP). In this section, we describe these formulations.

\subsection{Satisfiability Modulo Theories}

SMT can be seen as a generalized approach to Boolean satisfiability problems, increasing expressiveness by using e.g.\ first order logic instead of propositional logic. Another approach is to ask for satisfiability with respect to some background theory. This theory can fix interpretations of predicate and function symbols to be e.g.\ integers, reals, arithmetic, quantifiers or arrays \cite{biere2009}. Here, we use the core theory, which is defining the basic Boolean operators, and we have quantifiers available, which allows the following straight forward formulation of the CRR(\matrixS,\matrixR) problem:

\begin{equation}
\begin{aligned}
s_{ij} = 1 &\Leftrightarrow &\exists a: &(e_{ia} \wedge p_{aj}) &= \mbox{TRUE},&&\forall i,j =1,\dots,n\\
s_{ij} = 0 &\Leftrightarrow &\forall a: &(\lnot e_{ia} \vee \lnot p_{aj}) &= \mbox{TRUE},&&\forall i,j =1,\dots,n\\
r_{ab} = 1 &\Leftrightarrow &\exists i: &(e_{ib} \wedge p_{ai}) &= \mbox{TRUE},&&\forall a,b =1,\dots,m\\
r_{ab} = 0 &\Leftrightarrow &\forall i: &(\lnot e_{ib} \vee \lnot p_{ai}) &= \mbox{TRUE},&&\forall a,b =1,\dots,m
\end{aligned}
\label{bools}
\end{equation}

A formulation in propositional logic is also possible, but Eq.~(\ref{bools}) can be expressed directly in the SMT-LIB language \cite{smtLibFormat10}, which is the standard description language for input for SMT-solvers. We formalized \matrixS, \matrixR, \matrixE\ and \matrixP\ as uninterpreted functions, which take as arguments species and reactions, and give as result a Boolean value. 
Thus, in the SMT formulation we have 4 uninterpreted functions (one for each matrix), we declare two data types (species and reactions) and we have \(n+m\) variables of the mentioned data types.

\subsection{Boolean Logic}

SAT-solvers for Boolean satisfiability problems do not allow first order logic, thus Eq.~(\ref{bools}) is translated into propositional order formulae. Furthermore, the standard input format is the DIMACS format, which comprises the Boolean formula in Conjunctive Normal Form (CNF). The Boolean formula derived from Eq.~(\ref{bools}) is not in CNF and must therefore be converted. A plain conversion leads in worst case to a formula with exponentially many clauses, but methods like Tseitin encoding create an equisatisfiable formula in CNF, i.e., new variables are introduced and the new formula is satisfiable iff the original input formula is satisfiable \cite{tseitin1968,tseitin1983}. This new formula grows only linearly in terms of variables and clauses relative to the input formula. For this conversion, we use the tool \emph{limboole} \cite{limboole}. The number of variables for the original SAT formulation is in \(\BigO{nm}\), since only entries of \matrixE\ and \matrixP\ appear in the Boolean formula. By transforming the problem into CNF-form, new variables are introduced. An overview of the numbers of variables in the empirical study is shown in Table~\ref{varclaus} and discussed in Sec.~\ref{subsubempVars}.

\subsection{Integer Linear Programming}

Our third approach to solve the CRR(\matrixS,\matrixR) problem is to use Integer Linear Programming (ILP). Here, the quadratic program of Boolean matrix multiplication has to be linearized over the domain \{0,1\}. The reformulation of the constraints of matrix \matrixS\ is shown in Eqns.~(\ref{ilp0}-\ref{ilp14}), the reformulation for matrix \matrixR\ is carried out analogously. The reformulation follows standard techniques for linearizing quadratic programs to ILP. As described in Sec.~\ref{submatmat}, helper variables \(w_{ija}^S\) were additionally introduced to describe the existence of witness reactions \(a\) for entries \(s_{ij} = 1\) (resp. variables \(w_{abi}^R\) for witness species for entries  \(r_{ab} = 1\)). The objective function of the ILP is not of importance, since we only want to check for satisfiability. 
The number of variables for an ILP formulation is \(\BigO{nm+ n^2m + m^2n}\), where the first term gives the number of entries in \matrixE\ and \matrixP, and the second and third term give the variables \(w_{ija}^S\) and \(w_{abi}^R\). Empirical numbers of variables from the experiments are presented in Table~\ref{varclaus} and discussed in Sec.~\ref{subsubempVars}.

\begin{eqnarray}
    \min        & 0 &\nonumber\\
   \mbox{such that}&&\nonumber\\
    \forall s_{ij}=0: & e_{ia} + p_{aj} \leq 1 &\forall a= 1,\dots,m \label{ilp0}\\
    \forall s_{ij}=1: & \sum \limits_{a=1}^{m} e_{ia} \cdot p_{aj} \geq 1&\label{ilp1}\\
    	&    \Leftrightarrow&\nonumber \\
	 &\sum \limits_{a=1}^{m}  w_{ija}^S \geq 1& \label{ilp11}\\
	 &w_{ija}^S \leq e_{ia} &\label{ilp12}\\
	 &w_{ija}^S \leq p_{aj}&\label{ilp13}\\
	 &w_{ija}^S \geq e_{ia}+p_{aj}-1&  \forall a= 1,\dots,m \label{ilp14}
       \label{ilpform}
\end{eqnarray}

\section{Empirical Section}
\label{secEmp}

In this section, we first apply different solvers on randomly created
instances to get a better understanding of the hardness of the problem. We
then compare instances of \matrixS\ and \matrixR\ from real-world reaction
networks with random instances of comparable size.

\subsection{Computational Limits with Random Test Cases}
\label{subsecRand}

To examine the solvers' computational limits on the CRR(\matrixS, \matrixR)
problem, we apply them on randomly generated instances of \matrixS\ and
\matrixR, which we constructed in the following way.

\subsubsection{Test Data Set}
\label{subsubtestgen}

The test set was created with two pairs of parameters: The pair
\emph{(n,m)} of sizes of \matrixS\ (\(n \times n\)) and \matrixR\ (\(m
\times m\)), which were chosen as \((n,m) \in \{ (10,10), (20,10), (20,20),
(40,20),\) \( (40,40)\}\).

The second pair of parameters is \emph{(p,q)}, which defines the proportion
\(p\) (resp. \(q\)) of zeros out of all entries in \matrixS\
(resp. \matrixR), \(p,q \in [0,1]\). For each pair of matrices \matrixS\
and \matrixR, $p$ and $q$ were chosen uniformly at random in
\([0,1]\). According to the parameters \((p,q)\), the respective number of
zero entries were determined, and their positions in the matrices \matrixS\
and \matrixR\ were chosen uniformly at random.

\subsubsection{Experiments}

The experiments were run on Intel(R) Core(TM) i5 CPU 650 @ 3.20GHz
machines. We used state-of-the-art solvers for all declarative approaches,
namely the SMT-solver \emph{Z3} \cite{de2008z3}, the SAT-solvers
\emph{MiniSAT} \cite{een2004,sorensson2005} and \emph{lingeling}
\cite{lingeling2010}, and the ILP-solver IBM ILOG \emph{CPLEX} Optimization
Studio 12.5.

We ran the different solvers each on 1000 instances of each of the 5
different matrix sizes. In case of a successful reconstruction (i.e., a
pair \matrixE\ and \matrixP\ can be found from the given \matrixS\ and
\matrixR), an instance is called ``satisfiable''. It is called
``unsatisfiable'' if no reconstruction exists. In both cases the instances
are called ``solvable''. If the solving time exceeded 3600 seconds, the
process was aborted for that instance, which is then called
``indetermined''.

\subsubsection{Number of Variables}
\label{subsubempVars}

By converting a Boolean formula into CNF using, e.g., the Tseitin
method \cite{tseitin1968,tseitin1983}, the number of variables in CNF
formulation grows only linearly. The number of variables in the
original formula is \(\BigO{nm}\), where we observed approximately
\((n+m)\cdot2nm\) variables in the CNF formula of the test cases
(Table~\ref{varclaus}). Striking is the low variance in the number of
variables, which shows that the number of zeros and ones have only a
small impact on the size of the CNF.

As mentioned earlier, the number of variables for an ILP formulation is
\(\BigO{nm+ n^2m + m^2n}\). Empirical numbers of variables from the
experiments reflect the influence of the number of ones in the matrices
\matrixS\ and \matrixR\ on the number of used variables. The minimal value
in number of variables is seen for a low number of ones in the matrices,
since a low number of witness variables \(w_{ijk}^S\) and \(w_{jik}^R\) are
needed to be added. Vice versa, the maximal value is seen for a high number
of ones in the matrices.

\begin{table}[h!]
\centering
\caption{Numbers of variables 
  from the empirical studies, using test data as described in \ref{subsubempVars}. As expected, the number of variables in CNF formulation grows only linearly in terms of $n$ and $m$. The number of variables in the ILP formulation highly depends on the number of zeros in \matrixS\ and \matrixR, which determines the number of introduced witness variables.}
\begin{tabular}{c|ccc |ccc}
	Size&SAT&in&CNF&&ILP&\\
	&&\(\sharp\)variables&
	&&\(\sharp\)variables&\\ 
	&min&max&median&min&max&median\\ \hline
	(10,10)&4159&4399&4349&
	310&2170&1205\\
	(20,10)&12303&12799&12699&
	560&6340&3435\\
	(20,20)&32885&33599&33399&
	1180&16440&8920\\
	(40,20)&97424&99199&98799&
	2600&49240&25010\\
	(40,40)&260033&262399&261599&
	3600&128840&67840\\
\end{tabular}
\label{varclaus}
\end{table}

\subsubsection{Computational Limits on Random Instances}
\label{subsubperf}

Fig.~\ref{bars} shows the proportion of different outcomes among the 1000 given instances for each solver and each instance size: The number of satisfiable instances are colored green, and the unsatisfiable are colored gray. The number of indetermined instances, which could not be solved in the given time border of 3600 seconds, are marked in red. The SMT- and SAT-solvers behave similar w.r.t. to the number of solvable instances, whereas the ILP-solver CPLEX had more time-outs than the other two and hence appear to be the slowest of the methods. For the largest instance size \((n,m) = (40,40)\), the ILP-solver timed out in around 60~\% of all cases, instead of around 40~\% time-outs for the SMT- and SAT-solvers. 

To compare the solving methods in more detail and additionally to
investigate on which type of the instances the solving methods succeed
or fail, the solving time of the instances is investigated. For each
solver and instance size, Fig.~\ref{curveslog} shows for time~$t$ how
many instances have a solving time below~$t$. The ILP-solver has
markedly fewer solved instances for each~$t$ than the SMT- and
SAT-solvers. Additionally, the ILP-solver as well as the SAT-solver
\emph{lingeling} show a different solving behavior than the other two
solvers, in terms of a delay in providing solutions. This different
behavior occurs especially on the larger instances. That the SMT- and
SAT-solver MiniSAT behave similarly could indicate that Z3 is using a
propositional SAT-solver, as the input contains only Boolean variables.

The increase of indetermined instances correlated to the instance size (Fig.~\ref{bars}) and the increasing number of solvable instances over time (Fig.~\ref{curveslog}) brings up the question, which types of instances can be solved and which can not. In Fig.~ \ref{dots} we show whether an instance of a certain size with certain \((p,q)\)-values was solvable (green for satisfiable, red for unsatisfiable) or was indetermined (blue). Especially from the case of \((n,m) = (10,10)\) it becomes apparent, that a phase transition from satisfiable to unsatisfiable instances takes place, when the value of $p$ and $q$ crosses a certain level. Additionally, another observation is striking: If the value of \((p,q)\)-values of the matrices is in a certain range, the instances are hard to solve. The solvers timed out on instances with an amount of zeros in the vicinity of the phase transition, but did solve very dense and sparse instances. This observation shows that the hardest instances to solve appear close to the phase transition line from satisfiable to unsatisfiable instances. This intuitively makes good sense as it is harder for the solvers to find evidence of the outcome. Note, that in the top right corner there exist satisfiable instances among the unsatisfiable as well (cmp. Fig.~\ref{upperCorner}A). Especially, if an instance of \matrixS\ and \matrixR\ contains only zeros, then it is apparent from Eq.~\ref{bools}, that such an instance is an 2-SAT formula and additionally the CRR(\matrixS, \matrixR) problem is trivially satisfiable.

\subsection{Real World Instances}
\label{subSparse}

In this section, we compare random instances with real world networks. Reaction networks are usually large and sparse, thus they would belong in the upper right corner of a figure like Fig.~\ref{dots}. To compare real world networks with random instances and to compare the solvers' behavior on them, we selected test data in the following way.

\subsubsection{Test Data Sets}

We chose instances of reaction networks to be pathways of the yeast Saccharomyces cerevisiae (strain S288C), which are derived from the database \emph{MetExplore} \cite{metexplore2010}.
Table~\ref{varReal} shows this selection of small real world networks, where networks of size larger than \((n,m)=(130,130)\) had to be discarded, because the translating tool to CNF or the solvers themselves ran out of memory. 

Note, that in case of the real networks, the hypergraph (defined by \matrixE\ and \matrixP) is given and \matrixS\ and \matrixR\ are derived from that. Thus, the CRR(\matrixS, \matrixR) problem is satisfiable for all of these instances.

Table~\ref{varReal} also provides the \((p,q)\)-values of the real networks, where \(p\) and \(q\) were always greater than \(0.96\). Additionally, we created three sets of each 10000 random instances of comparable size, which we chose to be \((n,m) \in \{(40,40), (100,100), (120,120)\}\). To achieve as well a comparable sparsity to the real chemical networks, the values of $p$ and $q$ are uniformly at random chosen in \([0.9, 1]\). According to the p- and q-value, entries in the matrices \matrixS\ and \matrixR\ were chosen to be zero.

\subsubsection{Comparison of Random Instances with Real-world Networks}

As one can see in Table~\ref{varReal}, all sparse instances tested, real instances as well as random instances, are solvable with the SAT-solvers within seconds, in contrast to hard instances of smaller sizes. The fact that no solver timed out on any instance indicates a high impact of the sparsity of the matrices \matrixS\ and \matrixR\ has on the solvers' performances.

The solvabilty of these instances is depicted in Fig.~\ref{upperCorner}, where only the the upper right corner (i.e., instances with \((p,q)\)-values of (\(\geq0.9,\geq0.9\))) is shown. For comparison, the real networks are included in blue in Fig.~\ref{upperCorner}C. The fact that just a small fraction of randomly created instances of same sparseness as our real world instances is satisfiably solvable, whereas all \matrixS\ and \matrixR\ instances derived from real networks of course are, indicates that another structural property than sparsity alone characterizes real networks, e.g.\ clustering or connectedness.

\begin{table}
\caption{Some characteristic parameters of four metabolic pathways of Saccharomyces cerivisiae, and additionally its full reaction network are shown. A dash indicates an aborted run, due to memory issues. The networks are obtained from the MetExplore database \cite{metexplore2010}. Additionally, the parameters of random networks of different sizes are depicted. The numbers on variables and time in the random networks are the median over 10000 instances. The solver used is lingeling.}
\begin{center}
\begin{tabular}{ c|c|c|c|c}
Pathway&Size&CNF var&
time&(p,q)\\ \hline
4-Hydroxybenzoate&(110, 102)&4790939&
29.69&(0.979, 0.964) \\ \hline
TCA&(117, 111)&5961032&
33.77&(0.980, 0.961)\\ \hline
Sphingo Lipids&(125, 116)&7032499 & 
43.35&(0.982, 0.967)\\ \hline
Chorismate&(148, 143)&-&-&(0.984, 0.972)\\ \hline
Saccharomyces cerevisiae&(441, 504)&-&-&(0.993, 0.989) \\  \noalign{\hrule height 1.5pt}
Random Networks & (40,40) & 260951 & 
0.9 & (\(\geq0.9,\geq0.9\)) \\
Random Networks & (100,100) & 4030951 & 
16.44  & (\(\geq0.9,\geq0.9\)) \\
Random Networks & (120,120) & 6956569 & 
28.52 & (\(\geq0.9,\geq0.9\)) \\
\end{tabular}
\label{varReal}
\end{center}
\end{table}

\begin{figure}[t]
  \centering
  \subfloat[]{\label{fig:b1010b}\includegraphics[width=0.28\textwidth]{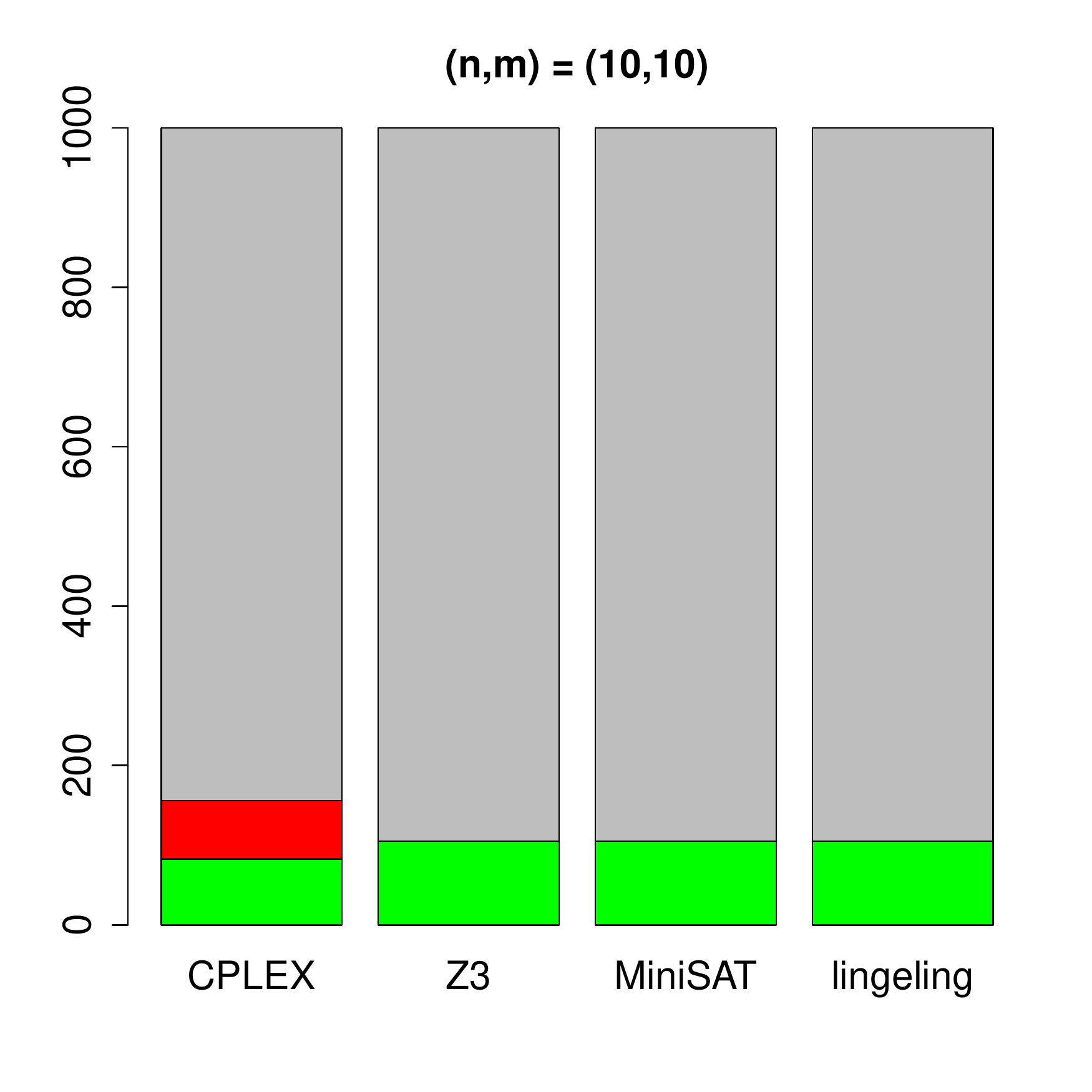}}
  \subfloat[]{\label{fig:b2010b}\includegraphics[width=0.28\textwidth]{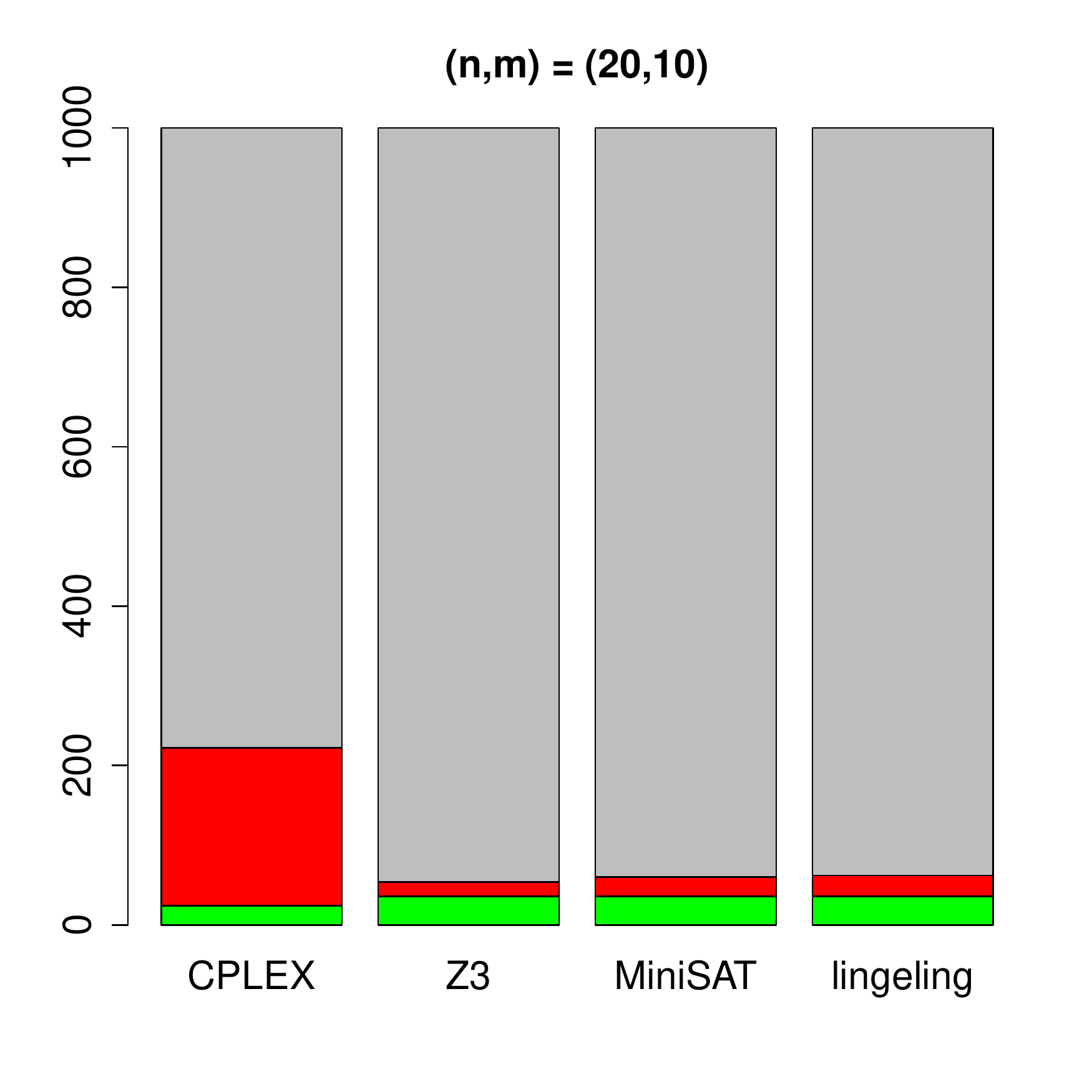}}
  \subfloat[]{\label{fig:b2020b}\includegraphics[width=0.28\textwidth]{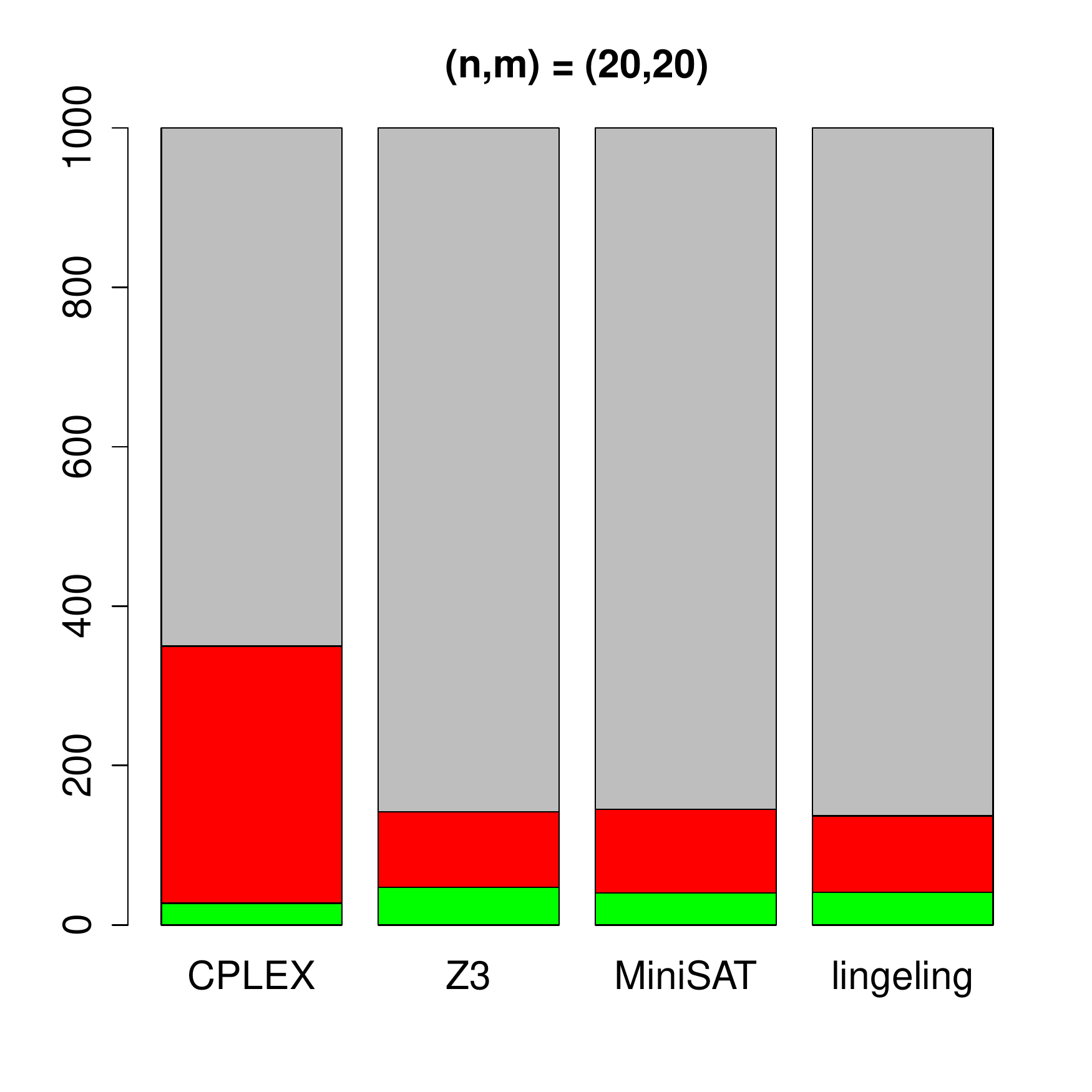}}\\
  \subfloat[]{\label{fig:b4020b}\includegraphics[width=0.28\textwidth]{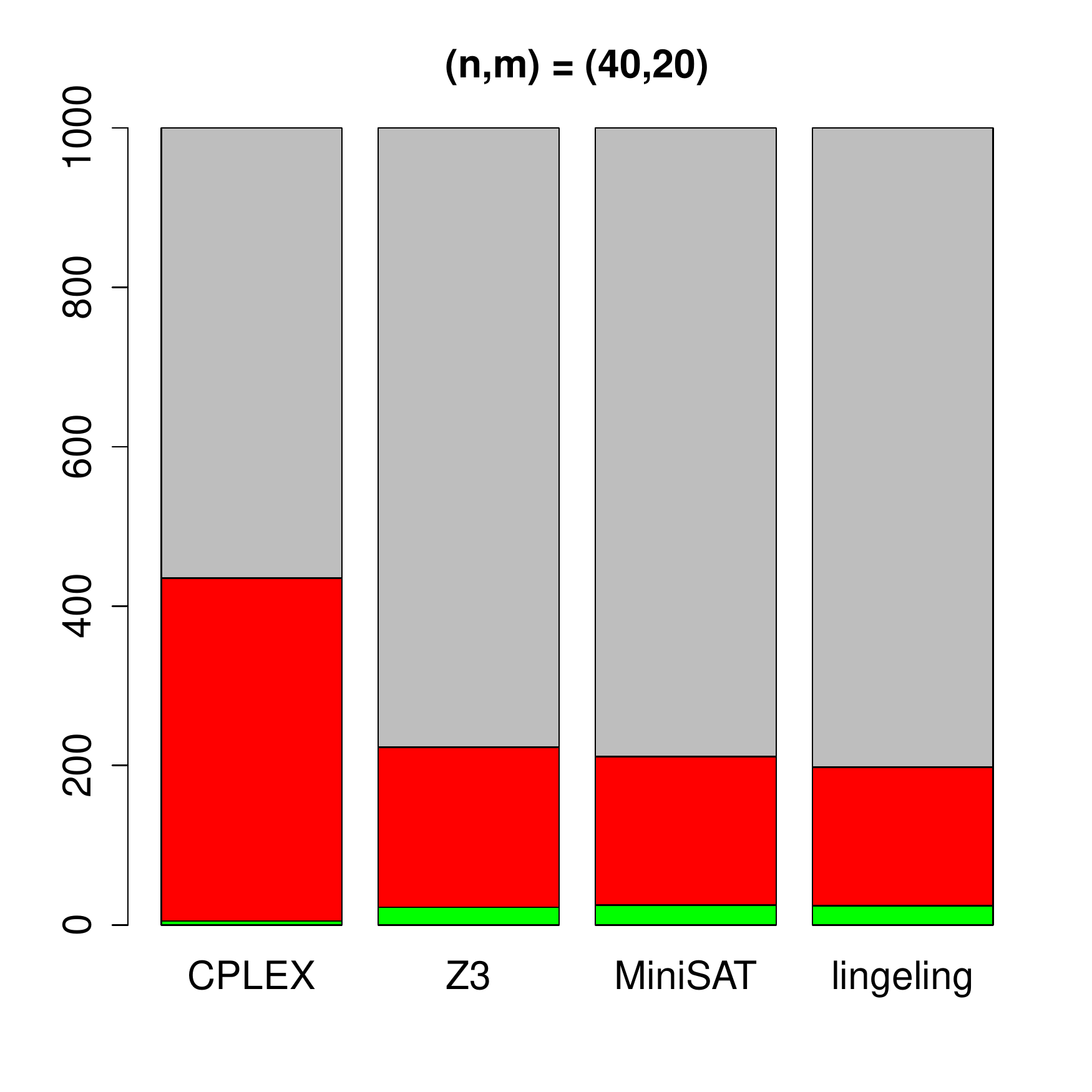}}
  \subfloat[]{\label{fig:b4040b}\includegraphics[width=0.28\textwidth]{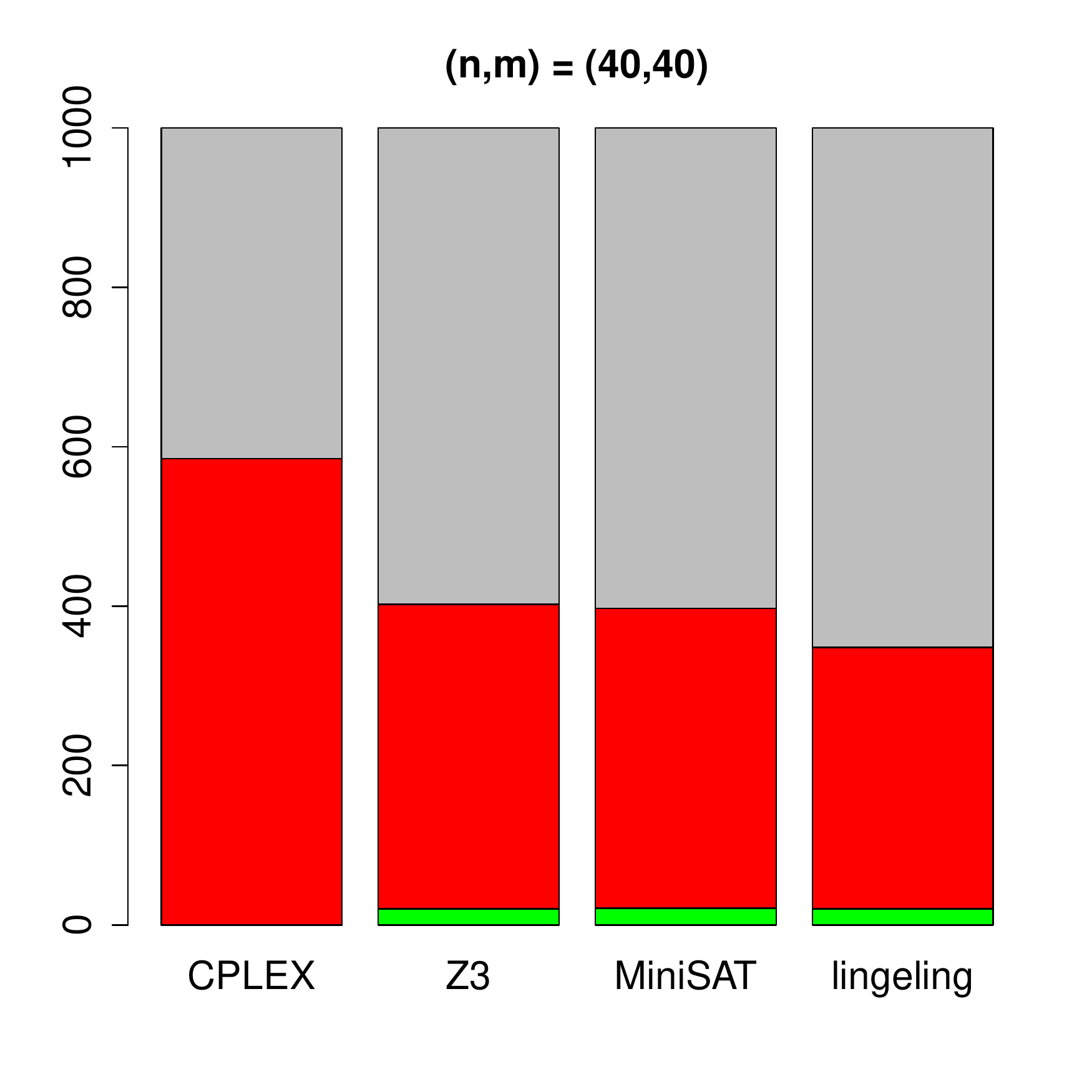}}\\
  \caption[Portions of sat/unsat/indet instances (of testruns)]{The bars show the portion of satisfiable (green), indetermined (red), and unsatisfiable (gray) instances among the 1000 given instances. They are arranged by the solving method.}
  \label{bars}
\end{figure}

\begin{figure}[t]
  \centering
  \subfloat[]{\label{fig:1010l}\includegraphics[width=0.28\textwidth]{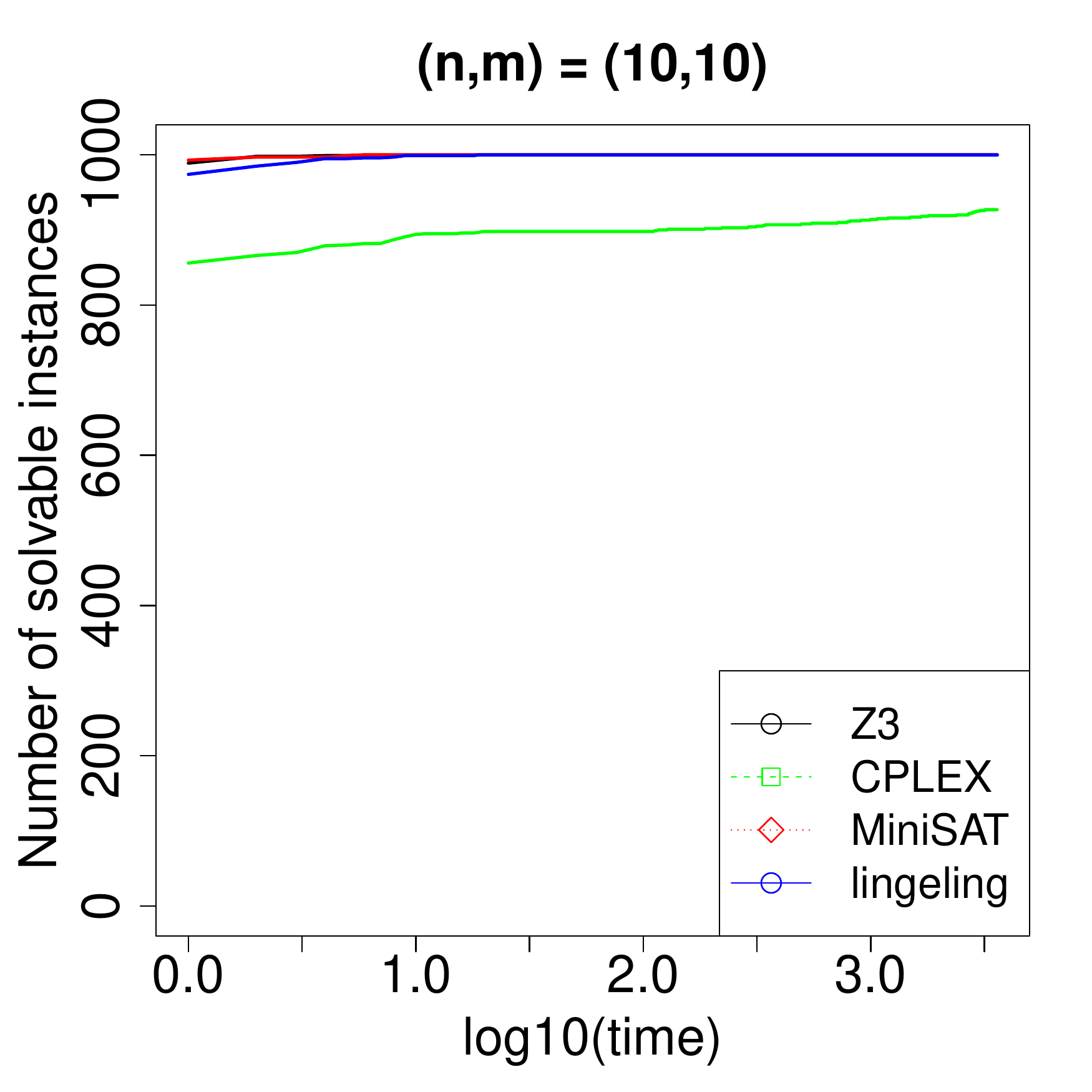}}
  \subfloat[]{\label{fig:2010l}\includegraphics[width=0.28\textwidth]{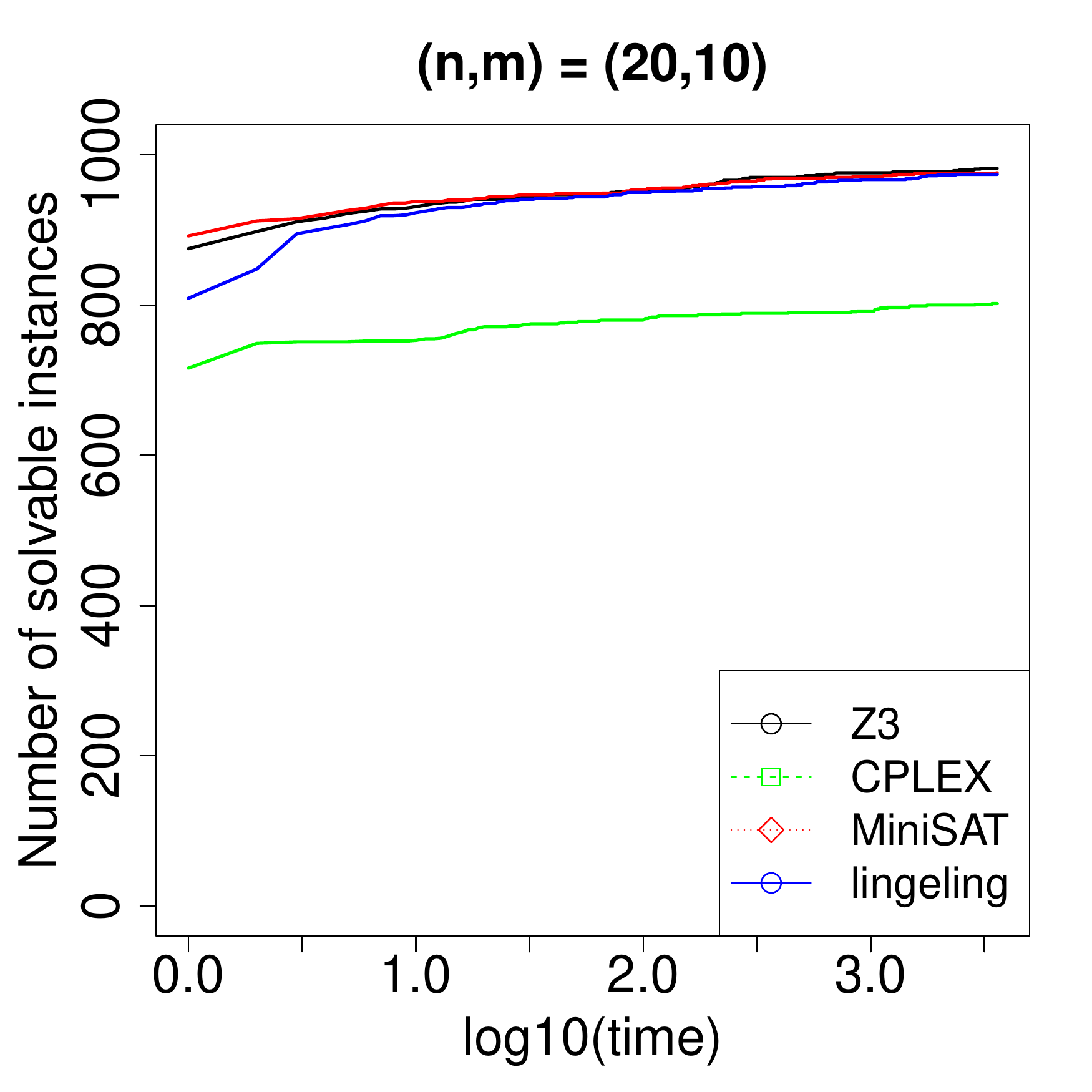}}
  \subfloat[]{\label{fig:2020l}\includegraphics[width=0.28\textwidth]{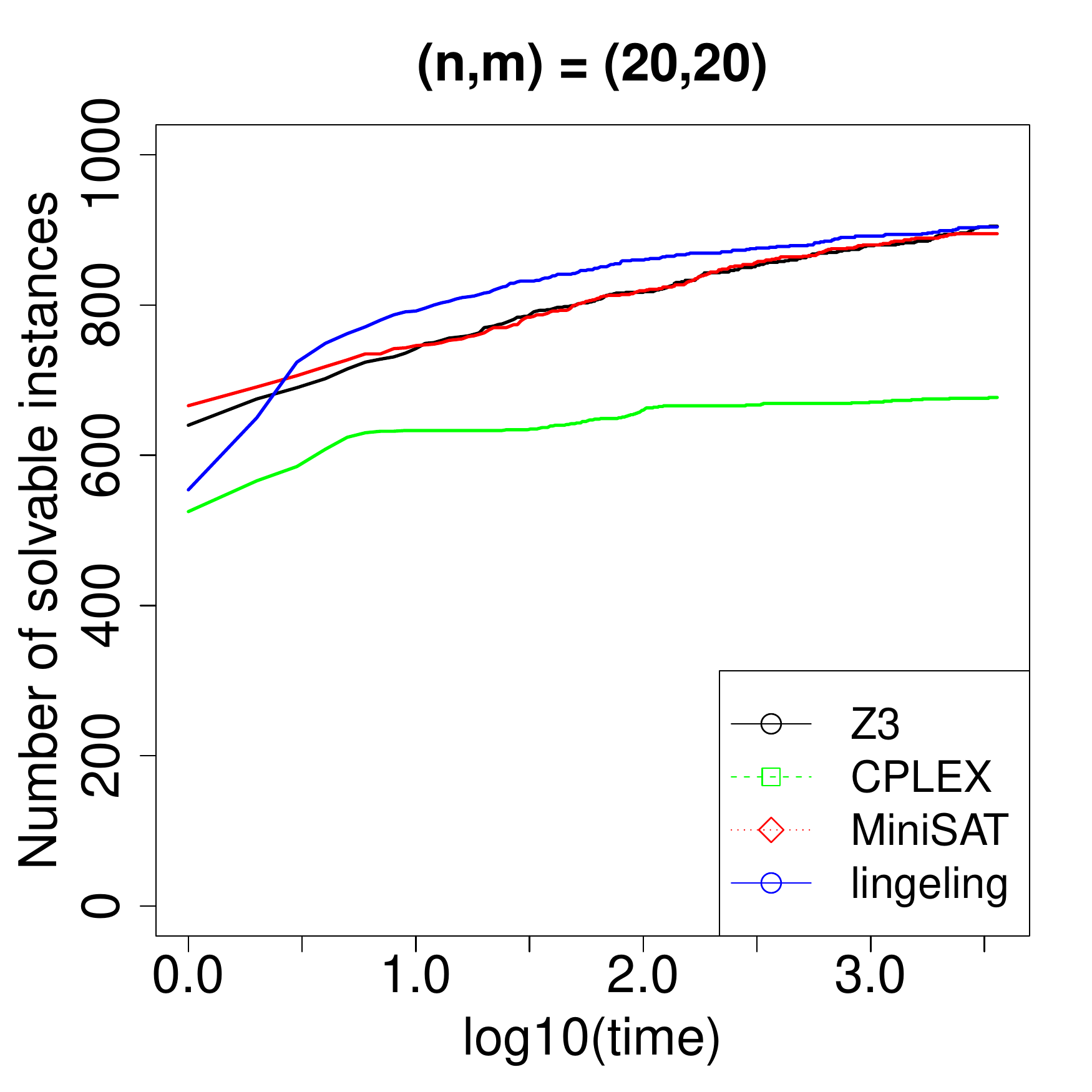}}\\
  \subfloat[]{\label{fig:4020l}\includegraphics[width=0.28\textwidth]{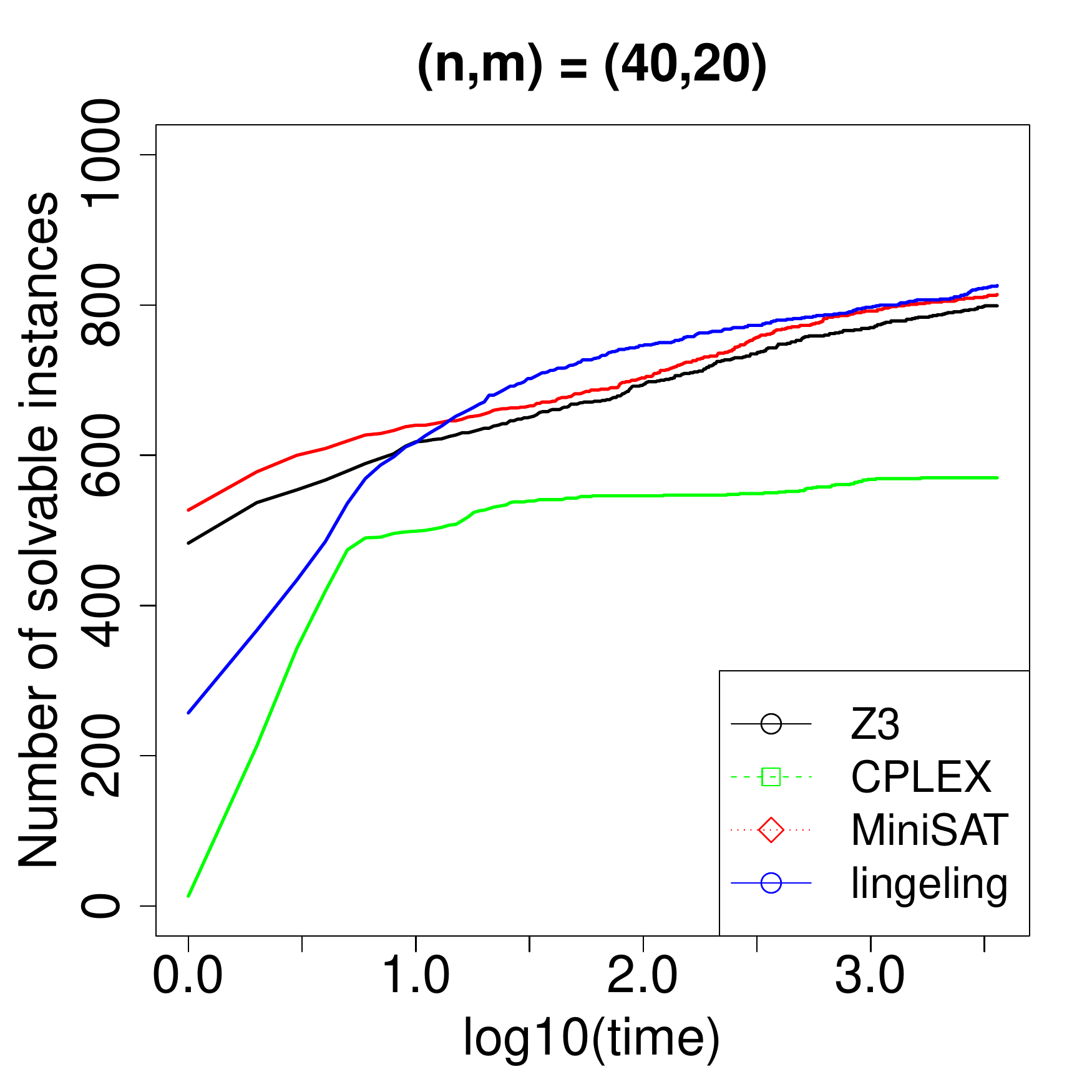}}
  \subfloat[]{\label{fig:4040l}\includegraphics[width=0.28\textwidth]{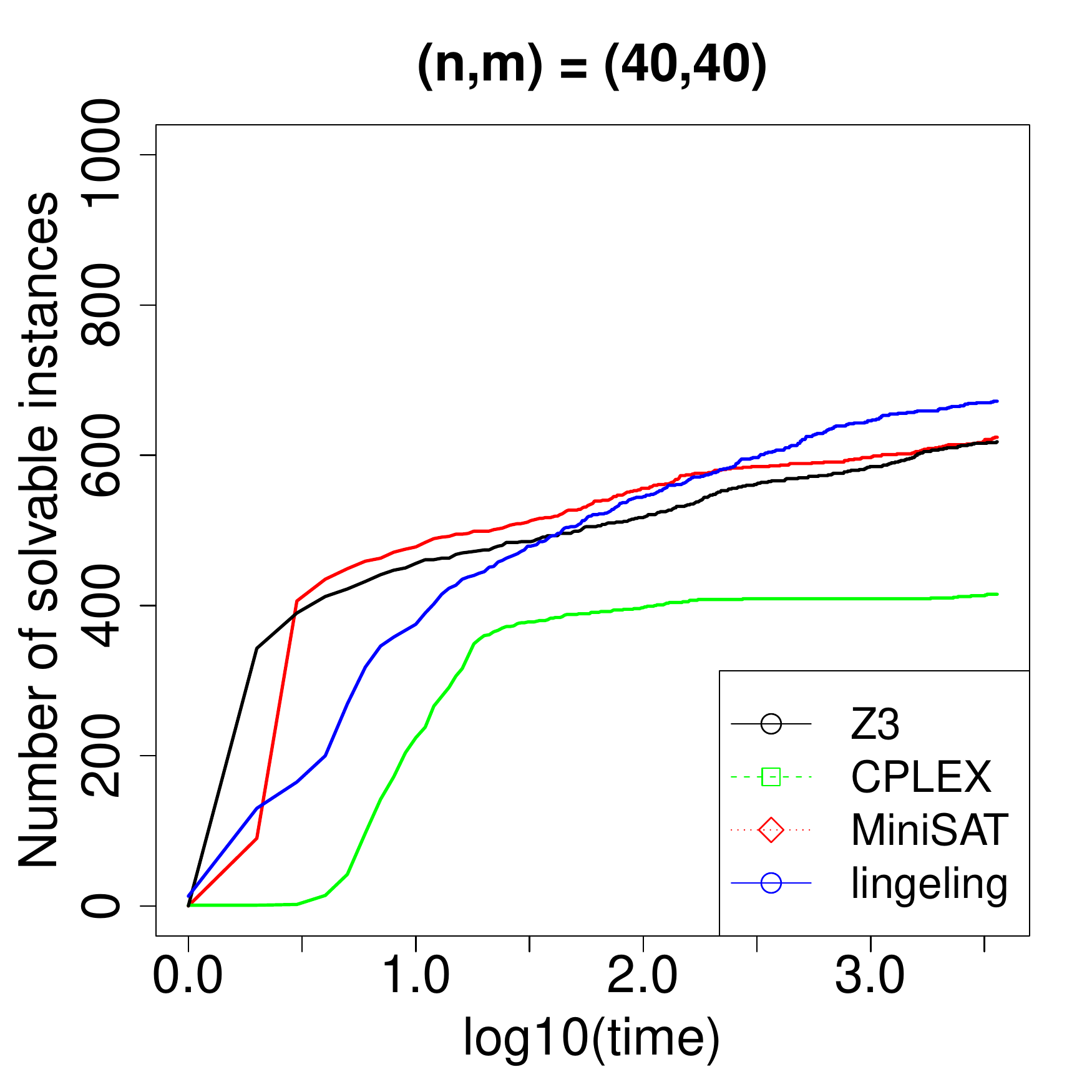}}
  \caption[Number of solvable instances over time]{Plots for the various instance sizes of how many of the 1000
  random instances each could be solved (as satisfiable or
  unsatisfiable) within in a certain amount of time. Depicted for each
  solver and over logarithmically scaled time axis. Time is measured in
  seconds.}
  \label{curveslog}
\end{figure}

\begin{figure}[t]
  \centering
\vspace{-8mm}
  \subfloat[]{\label{fig:d1010d}\includegraphics[width=0.45\textwidth]{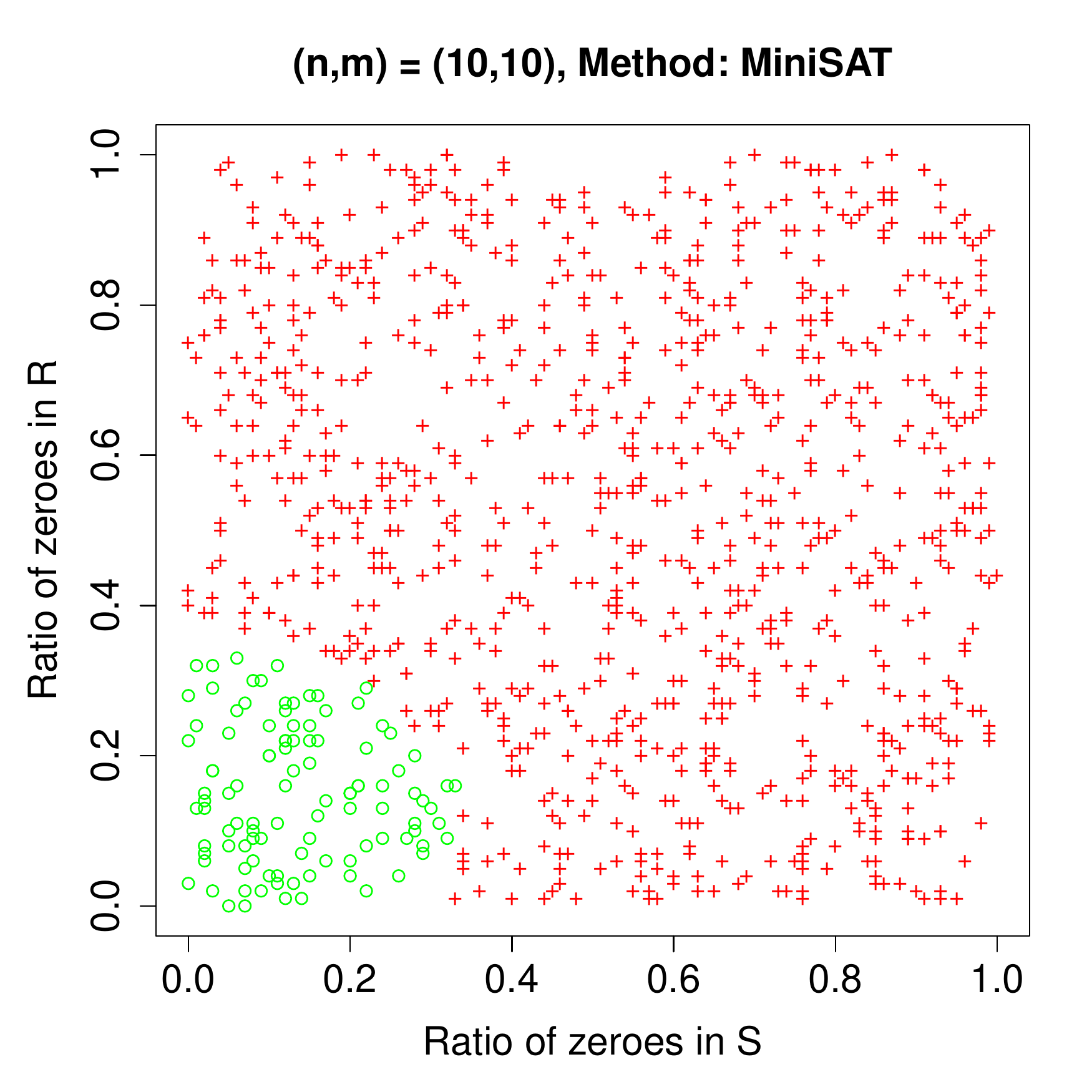}}
  \subfloat[]{\label{fig:d2010d}\includegraphics[width=0.45\textwidth]{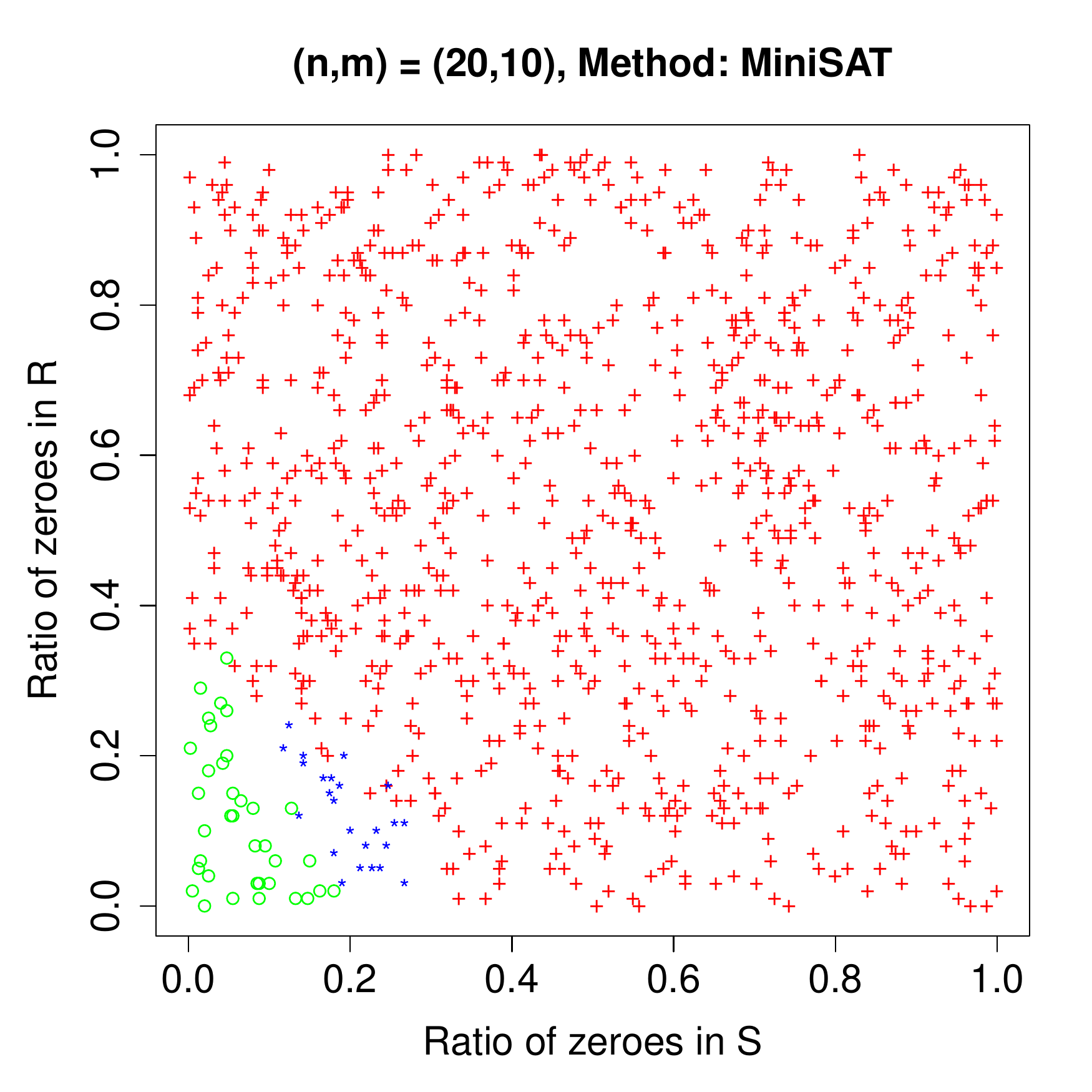}}\\
\vspace{-5mm}
  \subfloat[]{\label{fig:d2020d}\includegraphics[width=0.45\textwidth]{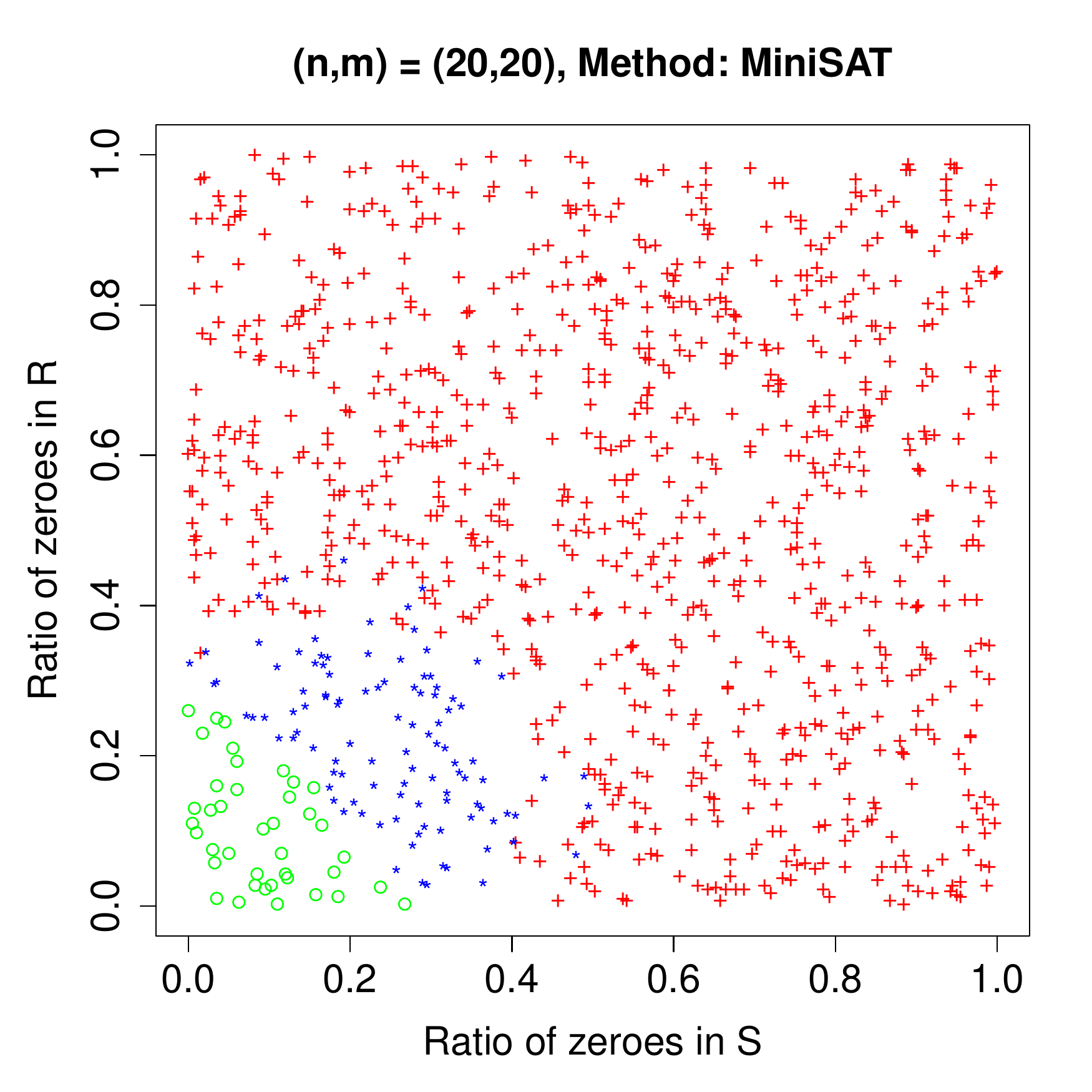}}
  \subfloat[]{\label{fig:d4020d}\includegraphics[width=0.45\textwidth]{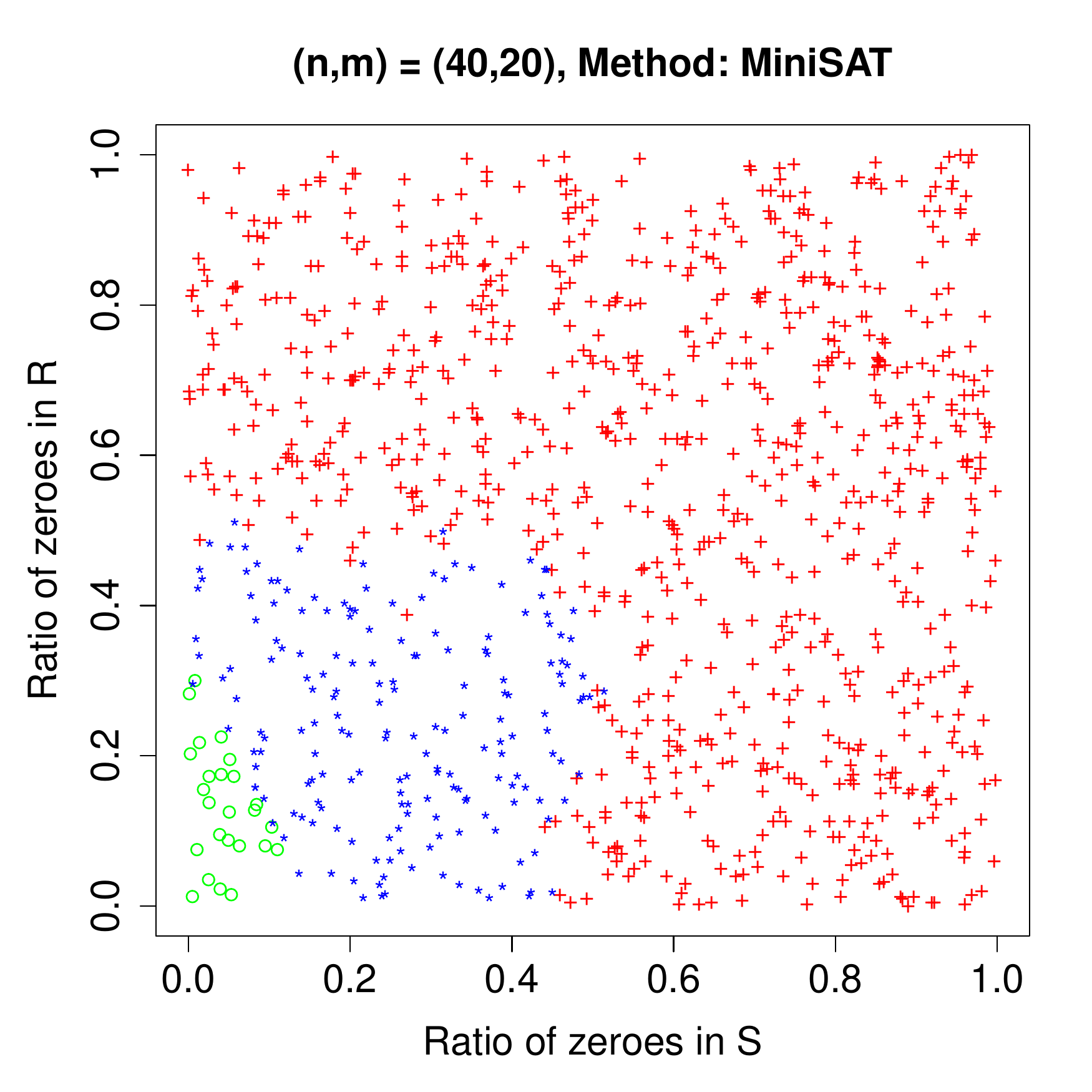}}\\
\vspace{-5mm}
  \subfloat[]{\label{fig:d4040d}\includegraphics[width=0.45\textwidth]{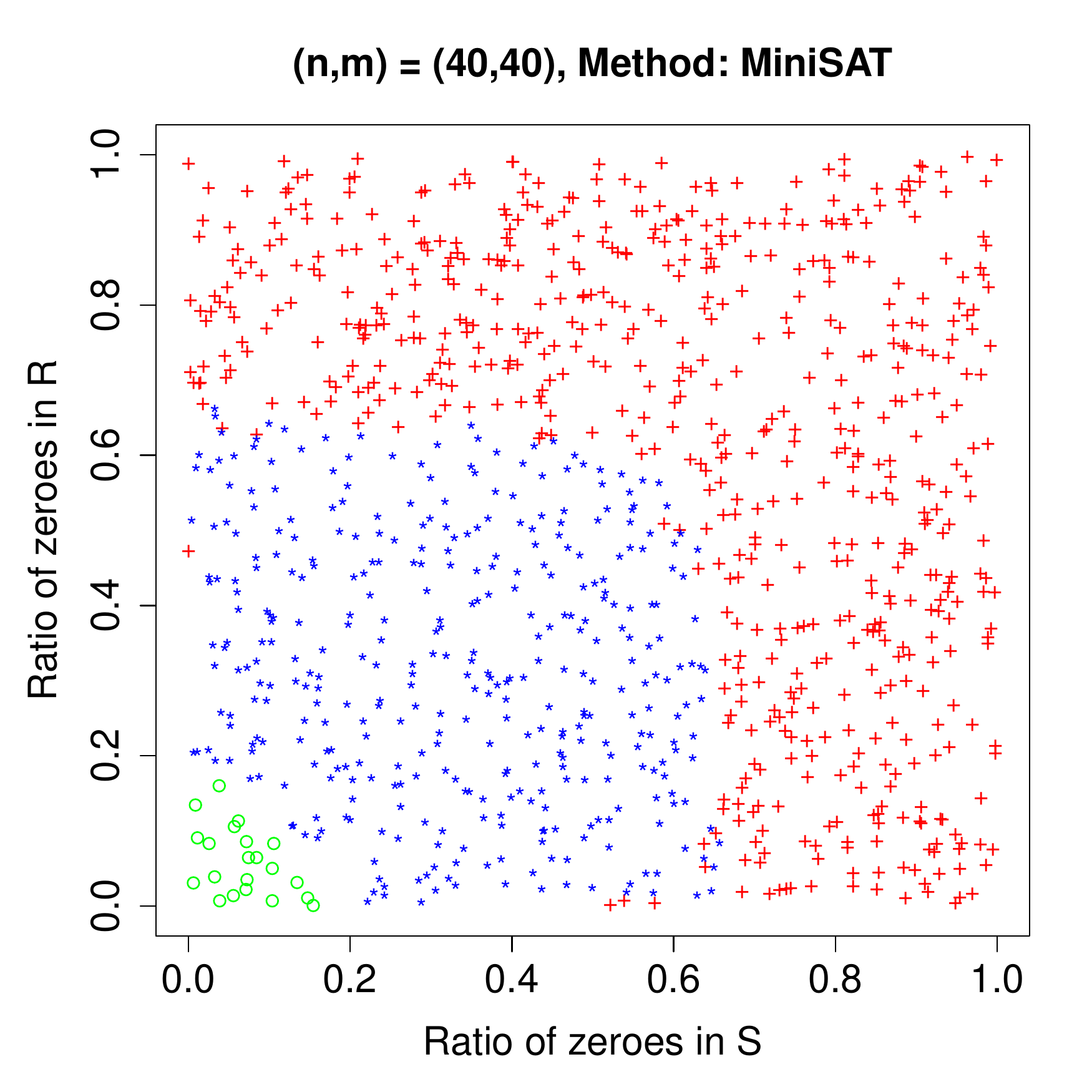}}\\
  \caption[Solvability status of test instances by $(p,q)$-value]{Each point marks the $(p,q)$-value of a random instance of a certain size. If the instance was satisfiable solvable, it is marked by a green circle, otherwise by a red cross. If the instance was not solvable in the given time of 3600 seconds, the instance is marked by a blue asterisk. A phase transition of easy and hard instance classes seems apparent.}
  \label{dots}
\end{figure}

\begin{figure}[t]
  \centering
\vspace{-8mm}
    \subfloat[]{\label{fig:d4040u}\includegraphics[width=0.45\textwidth]{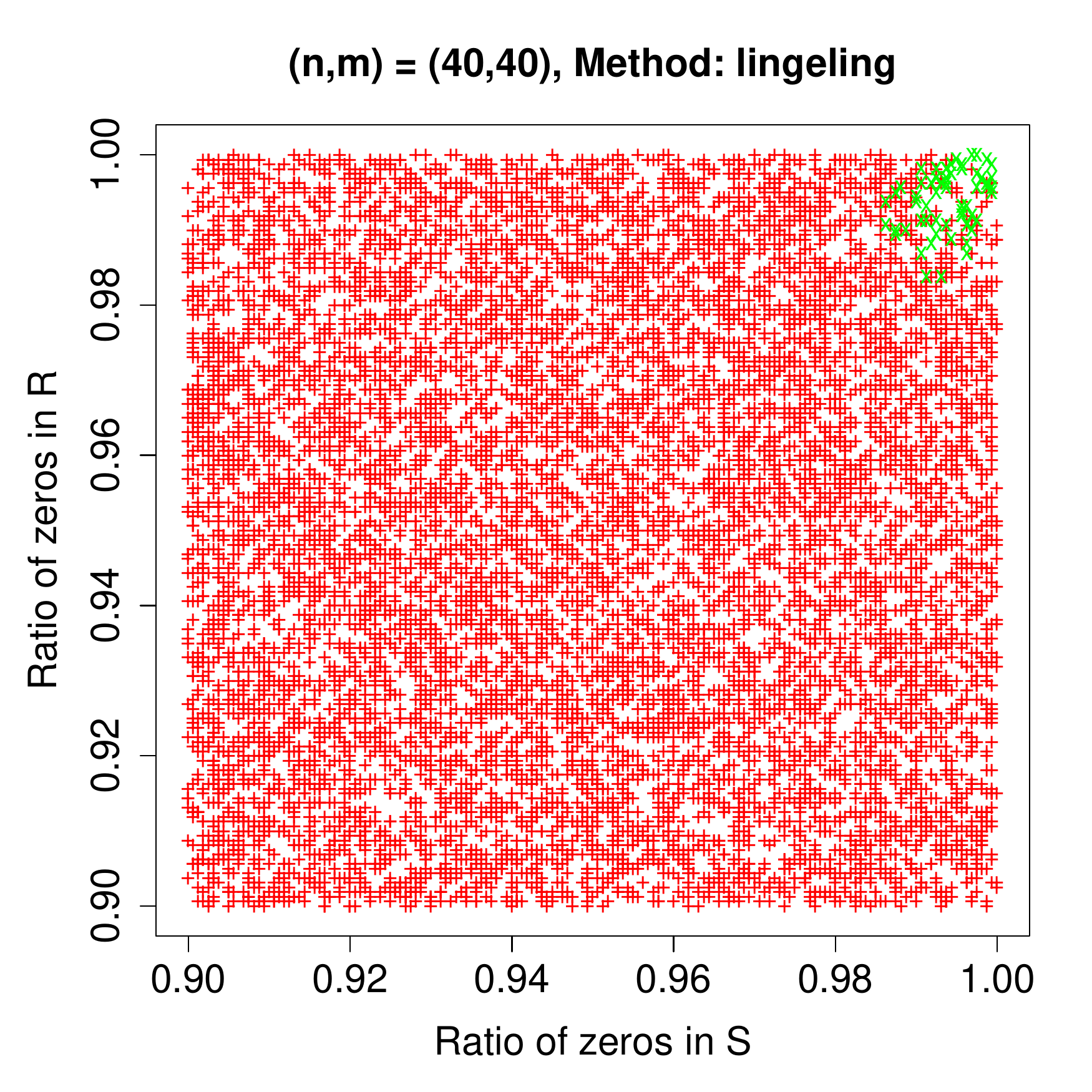}}\\
\vspace{-5mm}
  \subfloat[]{\label{fig:d100100u}\includegraphics[width=0.45\textwidth]{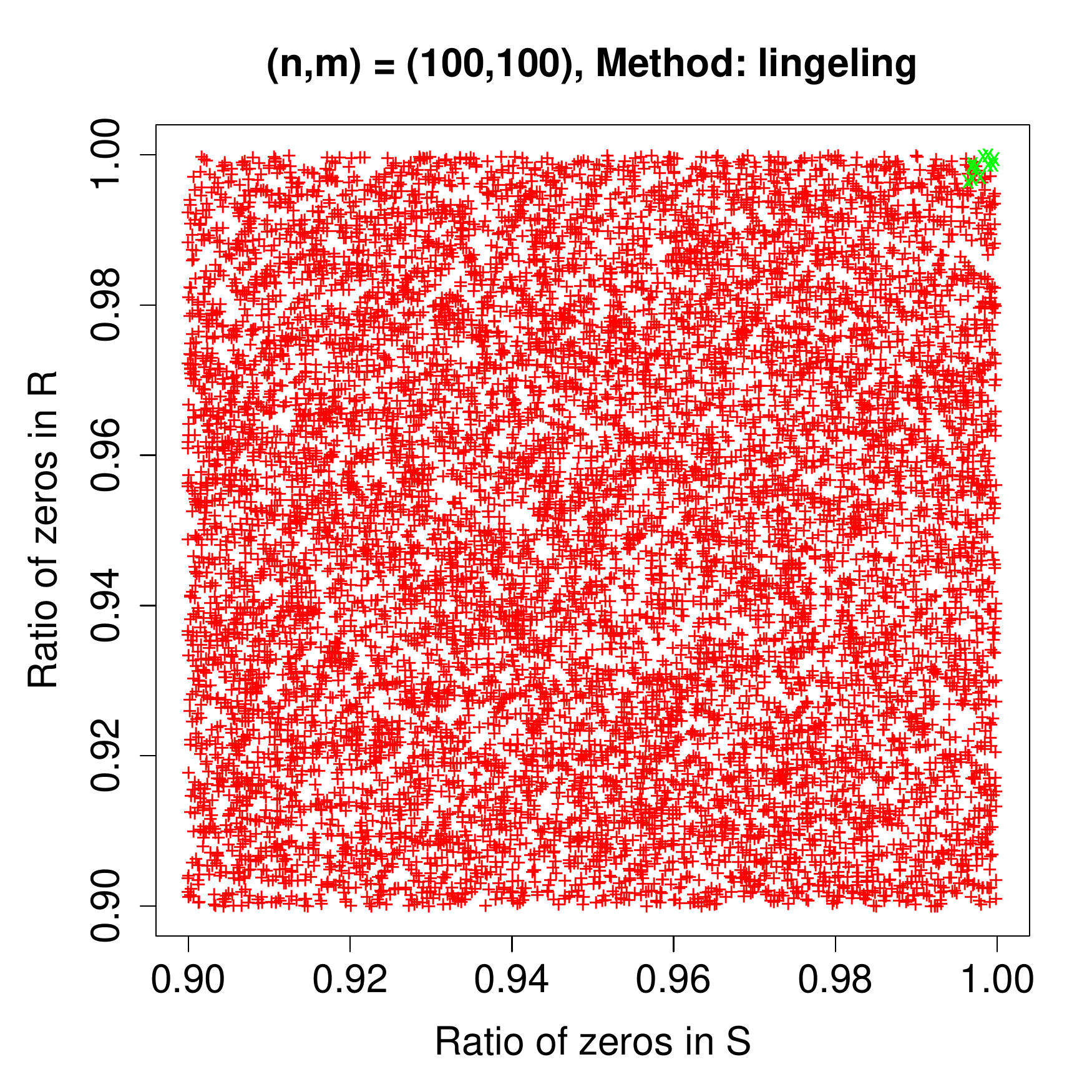}}\\
\vspace{-5mm}
    \subfloat[]{\label{fig:d120120u}\includegraphics[width=0.45\textwidth]{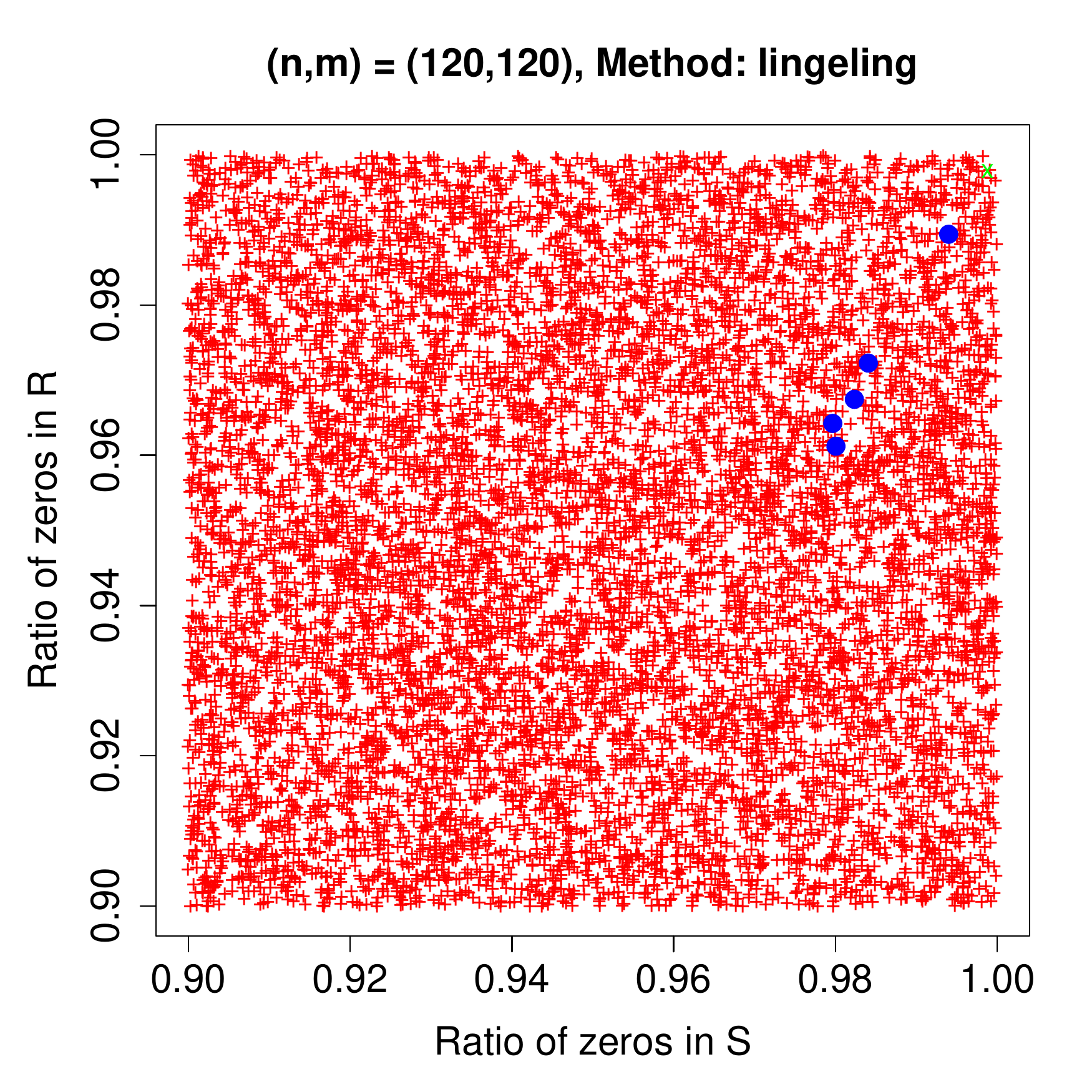}}
  \caption[Solvability status of sparse test instances by $(p,q)$-value]{The upper right corner of Fig.~\ref{dots} for certain sizes. The red crosses mark unsatisfiable instances, the green crosses satisfiable instances. Note that there only appear solvable random instances for these $(p,q)$-values greater than $(0.9,0.9)$, the solver lingeling did not time out on any instance. The big blue dots mark the $(p,q)$-values of the real world instances from Table~\ref{varReal}. The grid-like structure of Fig.~\ref{upperCorner}A follows from the instance's size and not from a lack of randomness.}
  \label{upperCorner}
\end{figure}


\section{Conclusion}
\label{secConc}

In this paper we addressed the problem of reconstructing a hypergraph from two given simple graphs. The problem, called the Compound-Reaction-Reconstruction (CRR(\matrixS, \matrixR)) Problem, is motivated by methods in chemistry, and models the reconstruction of a reaction network from the chemical \mbox{\matrixS-graph} and \mbox{\matrixR-graph}. Since these simple graphs are objects of analysis in chemistry, but contain only partial information about the underlying chemical network, this problem is of significant interest. As our first contribution, we proved this problem to be NP-complete.

As our second contribution, we empirically investigated the solvabilitiy
of the problem using standard declarative approaches. The results in
Sec.~\ref{subsecRand} show that for random instances we quickly get to
the computational limits with exact solving methods. It must be noted
that the size of the largest random instances can still be regarded as
quite small, since naturally occurring reaction networks normally
contain a larger number of species and reactions
(cmp. Table~\ref{varReal}). However, the computational limits proved to
be quite different for real world instances than for random
instances. This should not be very surprising, since e.g.\ the degree
distribution of nodes in real networks does not follow a uniform random
distribution, but tend to have certain structural properties. There is
an ongoing debate on suitable measures for similarity between random
graph models and natural reaction networks. There is agreement on the
modularity of the networks \cite{han2004}, but besides this, the
different modeling approaches focus on different measures. Popular
approaches to simulate natural reaction networks are Erd\H{o}s-R\'enyi
networks \cite{erdHos1960}, small-world graphs \cite{Wagner:01}, and
scale-free structures \cite{jeong2000} (which has also been subject of
criticism \cite{tanaka2005}). We refer to \cite{lacroix2008} for an
overview.

By comparing the solvability of random and real world instances, our
results indicate that sparsity seems to be a key property of natural
networks which allows them to be solved for larger instances, although
our data also indicate that this property alone is not sufficient for an
adequate characterization of natural networks.

A natural object odif further research is the characterization of network properties, in order to be able to sample real-world-like networks. Properties to investigate include the above mentioned scale-freeness of the directed hypergraph and whether it is transferable to the graphs of \matrixS\ and \matrixR, as well as parameters like clustering coefficient, depth, diameter, and connectedness.

Additionally, future work could include validation of how robust our statements in \ref{subSparse} are, regarding other simulated classes of graphs and further real world instances of reaction networks.


\bibliographystyle{abbrv}
\bibliography{DHGReconBin}

\clearpage
\section*{Appendix}

In this appendix, we provide an example of a pair of graphs \matrixS\
and \matrixR\ for which there exists no hypergraph having \matrixS\ and
\matrixR\ as its species and reaction graphs.

For the graphs \matrixS\ and \matrixR\ of Fig.~\ref{SRreconstruct}, it
is easily checked that the hypergraph in Fig.~\ref{fig:DHGsr} has these
as its species and reaction graphs. However, we now argue that if we change
the \matrixS\ and \matrixR\ pair by leaving out the red arc in
Fig.~\ref{SRreconstruct}, such a hypergraph no longer exists.

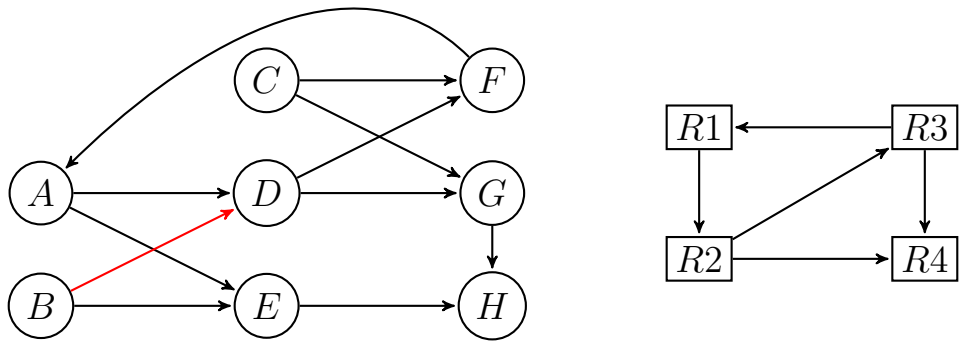
\begin{figure}[h!]
\centering
\begin{tikzpicture}[->,>=stealth',shorten >=1pt,auto,node distance=3cm,
  thick,species node/.style={draw,circle,font=\sffamily\Large\bfseries}]

\tikzstyle{reaction}=[draw,font=\sffamily\Large\bfseries]

\begin{scope}
\node[species node](A){$A$} ;
\node[species node](B) at ($(A)+(0,-1.5)$){$B$} ; 
\node[species node](C) at ($(A)+(3,1.5)$){$C$} ;   
\node[species node](D) at ($(C)+(0,-1.5)$){$D$} ; 
\node[species node](E) at ($(C)+(0,-3)$){$E$} ;   
\node[species node](F) at ($(C)+(3,0)$){$F$} ; 
\node[species node](G) at ($(F)+(0,-1.5)$){$G$} ;   
\node[species node](H) at ($(F)+(0,-3)$){$H$} ; 
\draw(A) to node{} (D);
\draw(A) to node{} (E);
\draw(B) [color=red] to node{} (D);
\draw(B) to node{} (E);
\draw(C) to node{} (F);
\draw(C) to node{} (G);
\draw(D) to node{} (F);
\draw(D) to node{} (G);
\draw(E) to node{} (H);
\draw(G) to node{} (H);
\draw(F) [out = 135, in = 45] to node{} (A);
\end{scope}

\begin{scope}[xshift=6.25cm,yshift=0.875cm]
\node[ reaction](R1)at ($(2.5,0)$){$R1$} ;
\node[ reaction](R2) at ($(R1)+(0,-1.75)$){$R2$} ; 
\node[ reaction](R3) at ($(R1)+(3,0)$){$R3$} ;   
\node[ reaction](R4) at ($(R3)+(0,-1.75)$){$R4$} ; 
\draw(R1) to node{} (R2);
\draw(R2) to node{} (R3);
\draw(R2) to node{} (R4);
\draw(R3) to node{} (R4);
\draw(R3) to node{} (R1);
\end{scope}
\end{tikzpicture}

\caption[\matrixS- and \matrixR-graph. One arc matters concerning the reconstructibility]{These graphs are the \matrixS-graph and \matrixR-graph of the
hypergraph of Fig.~\ref{fig:DHGsr}. However, removing the red arc~BD
from \matrixS\ gives a pair of graphs which is not the \matrixS-graph
and \matrixR-graph of any hypergraph.}
\label{SRreconstruct}
\end{figure}

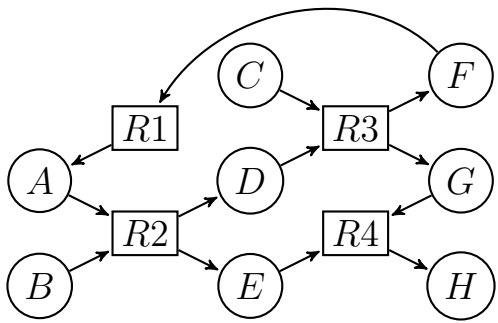
\begin{figure}[h!]
\centering
\begin{tikzpicture}[->,>=stealth',shorten >=1pt,auto,node distance=3cm,
  thick,species node/.style={draw,circle,font=\sffamily\Large\bfseries},scale=0.7]

\tikzstyle{reaction}=[draw,font=\sffamily\Large\bfseries]

\begin{scope}
\node[species node](A){$A$} ;
\node[species node](B) at ($(A)+(0,-2)$){$B$} ; 
\node[ reaction](R1)at ($(A)+(2,1)$){$R1$} ;
\node[ reaction](R2) at ($(R1)+(0,-2)$){$R2$} ; 
\node[species node](C) at ($(A)+(4,2)$){$C$} ;   
\node[species node](D) at ($(C)+(0,-2)$){$D$} ; 
\node[species node](E) at ($(C)+(0,-4)$){$E$} ;   
\node[ reaction](R3) at ($(C)+(2,-1)$){$R3$} ;   
\node[ reaction](R4) at ($(R3)+(0,-2)$){$R4$} ; 
\node[species node](F) at ($(C)+(4,0)$){$F$} ; 
\node[species node](G) at ($(F)+(0,-2)$){$G$} ;   
\node[species node](H) at ($(F)+(0,-4)$){$H$} ; 

\draw(A) to node{} (R2);
\draw(B) to node{} (R2);
\draw(R2) to node{} (D);
\draw(R2) to node{} (E);
\draw(C) to node{} (R3);
\draw(D) to node{} (R3);
\draw(R3) to node{} (G);
\draw(R3) to node{} (F);
\draw(E) to node{} (R4);
\draw(G) to node{} (R4);
\draw(R4) to node{} (H);
\draw(F) [out = 135, in = 60] to node{} (R1);
\draw(R1) to node{} (A);
\end{scope}

\end{tikzpicture}

\caption[Hypergraph, belonging to \matrixS\ and \matrixR\ in Fig.~\ref{SRreconstruct}]{A hypergraph having the graphs of Fig.~\ref{SRreconstruct} as
its \matrixS-graph and \matrixR-graph.}
\label{fig:DHGsr}
\end{figure}

So, assume a hypergraph~\matrixH\ exists having the graphs \matrixS\ and
\matrixR\ of Fig.~\ref{SRreconstruct} with the red arc~BD removed as its
species and reaction graphs.
Recall that in the species graph, two species $v_i$ and $v_j$ are adjacent
iff there is a reaction $a$ that has species $v_i$ as reactant and
species $v_j$ as product. Consequently, a hyperarc $a = (t(a),h(a))$ in
\matrixH\ induces in the species graph \matrixS\ of \matrixH\ all edges
of the complete directed bipartite graph between the vertex sets $t(a)$
and $h(a)$, and every edge of \matrixS\ is induced in this way by at
least one hyperarc.

Since no vertex in \matrixS\ has in- or out-degree greater than two, it
follows that for all hyperarcs~$a$ we have $|t(a)| \le 2$ and $|h(a)|
\le 2$. In particular, no hyperarc can induce more than four edges
of~\matrixS.
Also, since in \matrixS\ the out-degree of vertex~$F$ is one, it can
only be in $t(a)$ of hyperarcs~$a$ for which $|h(a)|$ is one. Similarly,
since the in-degree of vertex~$A$ is one, it can only be in $h(a)$ of
hyperarcs~$a$ for which $|t(a)|$ is one. It follows that the arc~$FA$
can only be induced by the hyperarc $(\{F\},\{A\})$, hence \matrixH\
must contain this.

Similarly, by the in- and out-degrees of vertices $E$, $G$, and $H$, the
arcs $EH$ and $GH$ must be induced by one or more of the hyperarcs
$(\{E\},\{H\})$, $(\{G\},\{H\})$ and $(\{E,G\},\{H\})$. Using two or
more hyperarcs for inducing $EH$ and $GH$ leaves at most one hyperarc to
induce the remaining seven arcs of \matrixS, as \matrixR\ shows that
there are exactly four hyperarcs in \matrixH. One hyperarc can induce at
most four arcs, so this is impossible. Hence, \matrixH\ must contain the
hyperarc $(\{E,G\},\{H\})$. Since seven arcs of \matrixS\ remains to be
induced, we must have $2 = |t(a)| = |h(a)| = |t(b)| = |h(b)|$ for the last two
hyperarcs $a$ and $b$. But this contradicts the outdegree of vertex~$B$
being only one.

\end{document}